\documentclass[floatfix,aps,prl,reprint,nofootinbib]{revtex4-2}

\usepackage{palatino}
\usepackage{graphicx}
\usepackage{morestyles}
\usepackage{physics}
\usepackage{mathrsfs} % \mathscr

\usepackage{xr-hyper} % for cross-document referencing
\usepackage[hidelinks]{hyperref} % for clickable links
\usepackage[capitalise]{cleveref} % for easy referencing
\Crefname{theorem}{Theorem}{Theorems}
\crefname{theorem}{Thm.}{Thms.}
\Crefname{figure}{Figure}{Figures}
\crefname{figure}{Fig.}{Figs.}
\Crefname{equation}{Equation}{Equations}
\crefname{equation}{Eq.}{Eqs.}
\Crefname{section}{Section}{Sections}
\crefname{section}{Sec.}{Secs.}
\usepackage{amsthm}
\newtheorem{corollary}{Corollary}
\newtheorem{theorem}{Theorem}
\newtheorem*{theorem*}{Theorem}
\newtheorem{definition}{Definition}

\newtheorem{proposition}{Proposition}

\usepackage{pdfpages}
\usepackage{pgffor}
\makeatletter
\AtBeginDocument{\let\LS@rot\@undefined}
\makeatother

\begin{document}

% \onecolumngrid

% \section{Todos}
% \begin{itemize}
%     \item Check if it's true that unital channels break up into Pauli errors followed by a unitary
%     \item At least draw the optimal phase diagram for random rotation errors
%     \item Redo analysis for separate X then Z errors and show it's sampling decoder
% \end{itemize}

% \clearpage
% \newpage
% \twocolumngrid

\title{Optimal recovery for quantum error correction}
\author{Sun Woo P. Kim}
\email{swk34@cantab.ac.uk}
\address{Department of Physics, King's College London, Strand, London WC2R 2LS, United Kingdom}
% However, a much broad set of operations could be used to recover. I'm gonna talk about not cookie cutter 
\begin{abstract}
The calculation of the error threshold of quantum error correcting codes typically proceeds as follows. First, syndromes are measured. Then, a decoder infers the error chain and the corresponding correction is applied. The threshold is then defined as the largest correctable error rate, with the maximum-likelihood decoder corresponding to the ``optimal'' threshold. However, a broader set of operations could be used to recover quantum information. The true optimal threshold should be optimised over all possible recovery schemes, which can be described by quantum channels.  Here, we study such optimal recovery channels and their thresholds $p_\mathrm{th}^\mathrm{opt}$. We introduce an information-theoretic quantity, mutual trace distance, which provides a necessary and sufficient diagnostic for sharply determining $p_\mathrm{th}^\mathrm{opt}$ without explicit optimisation. In contrast, previous works give a lower bound on $p_\mathrm{th}^\mathrm{opt}$ by specifying particular recovery schemes, e.g. Schumacher-Westmoreland (SW) which provides coherent information as a diagnostic to lower bound $p^\mathrm{opt}_\mathrm{th}$. We prove that the Petz and SW recovery schemes are optimal, i.e. their threshold is $p_\mathrm{th}^\mathrm{opt}$. With their optimality established, we explore the structure of optimal and non-optimal recovery schemes and their phase diagrams.
\end{abstract}

\maketitle

\paragraph*{Introduction ---}In most works, calculating the thresholds quantum error correcting stabiliser codes proceeds as follows. First, a noise channel $\mathcal{E}^{(p)}$ of strength $p$ is assumed. Then, the error correction procedure is given by the measurement of syndromes, followed by a scheme where, conditioned on the measurement outcomes, a decoder decoder infers the error chain and the corresponding correction is applied to hopefully revert back to the original state \cite{dennis2002topological,darmawan2017tensor,darmawan2018linear,bombin2012strong}. The threshold $p_\mathrm{th}$ then is where the failure rate goes from zero to a finite value as number of qubits $N \rightarrow \infty$.

For concreteness, consider the toric code affected by Pauli errors, such as bit-flip or phase-flip errors. Here, although there is no proof, it seems obvious that the above set of operations, combined with the maximum-likelihood (ML) decoder, which determines the most likeliy error, is `optimal' for error correction. However, the scheme of measuring syndromes then inferring matchings may not be the optimal procedure for error correction for other noise channels \cite{fletcher2007optimum}. 

Instead, consider a coherent error where each physical qubit is rotated by $U_\theta = e^{i \theta X}$. In this case, the optimal operation would be to simply undo the rotation. One may argue that, a priori, one cannot know by which angle $\theta$ the rotation will be applied to the system. However, without any characterisation of prior knowledge of errors, it is not possible to determine the theoretical value of the thereshold. Indeed, in classical Bayesian inference, one must assume a prior and likelihood distribution in order to compute the theoretical optimum threshold \cite{zdeborova2016statistical,p2025planted}. Similarly, in the quantum setting, one must assume an error channel (this could encode one's ignorance) to discuss the threshold.

A proper account for such coherent rotations would be to model the uncertainty in the rotation angle $\theta$ as a distribution $p(\theta)$. One can intepret this as a channel, which can then be written as a Pauli error channel $\mathcal{E}_\mathrm{Pauli}$ followed by a deterministic rotation $\mathcal{U}$, i.e. $\mathcal{E} = \mathcal{U} \mathcal{E}_\mathrm{Pauli}$~\footnote{This is also true for $U_\phi = e^{i \phi Z}$. Details are in Supplemental Material Sec. VI.}. % \cref{sm:sec:coherent-rotations}
Then, the optimal recovery channel should be to undo the rotation, $\mathcal{U}^{-1}$, followed by usual scheme of measuring the stabilisers and applying the correction as inferred by the ML decoder.

One can always convert the recovery problem into a classical Bayesian inference problem by assuming that measurement of stabilisers always follow the noise channel. Indeed, if this is assumed, the problem maps to a Bayesian inference problem \cite{kim2025measurement,behrends2025statistical}.
% ~\footnote{With the caveat that the inferred posterior distribution may contain the amplitude square of a complex partition function which may be intractable to compute in practice \cite{behrends2025statistical}.}

In this work, we consider the problem in its entirety, viewing the recovery operation, described as a channel $\mathcal{R}$, as part of the optimisation for the threshold. For a given a code $\mathcal{C}$ and noise channel $\mathcal{E}$, $\mathcal{R}$ aims to approximately recover noise-affected states within the code space of $\mathcal{C}$, such that
\begin{align*}
    \mathcal{R} \, \mathcal{E} \rho_{\mathcal{C}} \approx \rho_{\mathcal{C}}.
\end{align*}
We consider $\mathcal{R}^\mathrm{opt}$, the optimal channel for recovery with respect to some loss function $\mathscr{L}$, such as the entanglement infidelity $1 - F_\mathrm{e}$, and its associated threshold $p_\mathrm{th}^\mathrm{opt}$. Previous studies have focused on sufficient conditions for recovery. For example, if coherent information $I_\mathrm{c}$ is preserved, then recovery is possible \cite{schumacher2002approximate}. 
% However, this does not rule out the possibility that even if it is not preserved, recovery may still be possible by $\mathcal{R}^\mathrm{opt}$. 

Here, we provide a necessary and sufficient condition for recovery, by proving an upper bound for the entanglement fidelity of $\mathcal{R}^\mathrm{opt}$ in terms of an information-theoretic diagnostic we call \emph{mutual trace distance} $T_{A:B}$, see \cref{fig:summary}. 

We define recovery `schemes', a family of recovery channels over number of physical qudits $N$ and noise strength $p$. This allows us to simply prove statements regarding thresholds, such as that known recovery schemes --- Petz and Schumacher-Westmoreland (SW) --- are optimal, having the same threshold as $\mathcal{R}^\mathrm{opt}$. With their optimality established, we explore their structure and the operational content. For example, we show that for qubit CSS codes under amplitude damping errors, optimal recovery schemes measure the $Z$-type syndromes and correct them in the usual way, but \emph{not} for the $X$-type syndromes, instead applying a coherent correction operation. In analogue to optimal and non-optimal quantum and classical Bayesian inference problems, we introduce non-optimal quantum recovery problems and provide a constraint on the shape of such phase diagrams. Finally, we close with future research directions.

\begin{figure}
    \centering
    \includegraphics[width=0.95\linewidth]{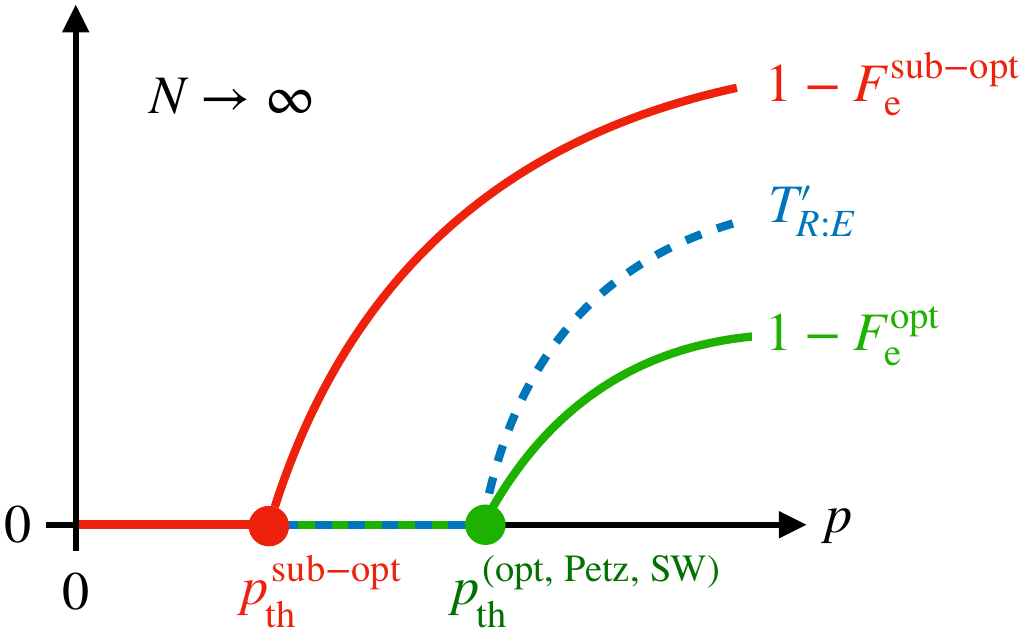}
    \caption{Schematic of some main results. First, we show that the mutual trace distance $T'_{R:E}$ (\cref{def:mutual-trace-distance}) is a necessary and sufficient quantity to detect the optimal threshold over all recovery schemes, $p^\mathrm{opt}_\mathrm{th}$, whose performance is measured by the entanglement infidelity $1 - F_\mathrm{e}$. Next, we show that the Petz and Schumacher-Westmoreland recovery schemes $\mathcal{R}^\mathrm{Petz}$, $\mathcal{R}^\mathrm{SW}$, are optimal in that their threshold exactly coincides with the optimal one, $p^\mathrm{opt}_\mathrm{th} = p^\mathrm{Petz}_\mathrm{th} = p^\mathrm{SW}_\mathrm{th}$. $p$ is the noise strength.}
    \label{fig:summary}
\end{figure}

\vspace{1.5em}
\paragraph{Optimal recovery channels and optimal thresholds ---}In order to discuss optimal recovery channels, one must consider the triple $(\mathcal{C}, \mathcal{E}, \mathscr{L})$. Here, $\mathcal{C}$ is a quantum error-correcting code, which could also be thought of as a preparation channel. $\mathcal{E}$ is an error channel, and $\mathscr{L}$ is a `loss function' used to measure the success ($\mathscr{L}=0$) of error-correction for some recovery channel $\mathcal{R}$. The optimal recovery channel is thus given by
\begin{definition}[Optimal recovery channel]
    \begin{align}
        \mathcal{R}^\mathrm{opt}_\mathscr{L} = \mathrm{argmin}_{\mathcal{R}} \mathscr{L}(\mathcal{R} \, \mathcal{E}; \mathcal{C}).
    \end{align}
\end{definition}

For the rest of this work, we will assume the loss function $\mathscr{L} = 1 - F_\mathrm{e}$, where the entanglement fidelity $F_\mathrm{e}$ is
\begin{align}
    F_\mathrm{e}(\mathcal{R} \mathcal{E}) = F\left[ \mathcal{R}_Q \mathcal{E}_Q \Phi_{QR}, \Phi_{QR} \right],
\end{align}
and $\Phi_{QR} = \ketbra{\Phi_{QR}}{\Phi_{QR}}$ is the density matrix for the pure entangled state $\ket{\Phi_{RQ}} = \sum_{k=0}^{d-1} \sqrt{\lambda_k} \ket{\overline{k}}_R \ket{k}_R$ between encoded qudit $Q$ with logical states $\ket{\bar k}$ and reference $R$. Importantly, the recovery channel can \emph{only act on} $Q$. $F[\rho, \sigma] = \norm{\sqrt{\rho} \sqrt{\sigma}}_1$ is the fidelity, equal to $F[\rho, \psi] = \sqrt{\mel{\psi}{\rho}{\psi}}$ if one of the states is pure. For the rest of the manuscript, we choose $\lambda_k = 1/d$, i.e. a Bell pair, where $d$ is the logical qudit dimension.

For this loss function, a recipe for constructing the optimal recovery channel $\mathcal{R}^\mathrm{opt}$ was shown in Ref. \onlinecite{fletcher2007optimum}. Unfortunately, its analysis is intractable as obtaining it requires semi-definite programming. An alternative recovery channel is the Petz recovery channel $\mathcal{R}^\mathrm{Petz}$, also known as the transpose recovery channel, detailed in the Supplemental Material (SM) Sec. III. % \cref{sm:sec:petz}.
A remarkable feature of this recovery channel is that it is `near-optimal', in the sense that it two-way bounds the entanglement fidelity of the \emph{optimal} recovery channel \cite{barnum2002reversing,zheng2024near},
\begin{theorem*}[Two-way bound of the Petz recovery channel, Theorem 2 of Ref. \onlinecite{barnum2002reversing}]
    \begin{align} \label{eq:petz-opt-two-way-bound}
        (F^\mathrm{opt}_\mathrm{e})^2 \leq F^\mathrm{Petz}_\mathrm{e} \leq F^\mathrm{opt}_\mathrm{e}.
    \end{align}
\end{theorem*}
Here we used $F^\mathrm{opt}_\mathrm{e} = F_\mathrm{e}(\mathcal{R}^\mathrm{opt} \mathcal{E})$ as a shorthand, and similarly for $F^\mathrm{Petz}_\mathrm{e}$.

To discuss thresholds and optimal thresholds, we consider a family of error-correcting codes $(\mathcal{C}^{(N)})_N$ parameterised by the number of physical qudits $N$, and a family of error channels $(\mathcal{E}^{(N, p)})_{N, p}$ parameterised by both $N$ and strength $p$. For example, such error channels could be an error channel separable on each physical qudit $\mathcal{E}^{(N, p)} = \prod_{i=1}^N \mathcal{E}_i^{(p)}$. Given this, we choose a family of recovery channels $(\mathcal{R}^{(N,p)})_{N, p}$ which we will refer to as a \emph{recovery scheme}. Then, the threshold $p_\mathrm{th}$ is defined as
\begin{definition}[Recovery scheme and threshold]
    Given the family of quadruples $(\mathcal{C}, \mathcal{E}, \mathscr{L}, \mathcal{R})_{N, p}$, where $(\mathcal{R})_{N, p}$ is referred to as a recovery scheme, the threshold defined as
    \begin{align}
        p_\mathrm{th} = \max_p \left\{p : \lim_{N \rightarrow \infty} \mathscr{L}(\mathcal{R}^{(N,p)} \mathcal{E}^{(N,p)}; \mathcal{C}^{(N)}) = 0 \right\}.
    \end{align}
\end{definition}

The optimal threshold then which is the theoretical maximum for any recovery scheme, is given by that of the optimal family of recovery channels $p_\mathrm{th}^\mathrm{opt} = p_\mathrm{th}(\mathcal{C}, \mathcal{E}, \mathscr{L}, \mathcal{R}^\mathrm{opt})_{N, p}$. Here, we assumed that $\mathscr{L}$ is monotonically non-decreasing with $p$~\footnote{ We could easily generalise to more parameters $\vec{p}$ and non-monotonic $\mathscr{L}$ by considering regions $\mathrm{QEC} = \{\vec p : \lim_{N \rightarrow \infty} \mathscr{L}(\vec p) = 0\}$ and $\mathrm{QEC}^c = \{ \vec p : \lim_{N \rightarrow \infty} \mathscr{L}(\vec p) > 0\}$. Then the thresholds $\vec{p}_\mathrm{th}$ would be determined by the boundary between such regions.}.

Note that the thresholds for different choices of loss functions are often the same. For example, for $\mathscr{L} = T$, where $T[\rho, \sigma] = \frac12 \norm{\rho - \sigma}_1$ is the trace distance, this follows from the two-way bounds between trace distance and fidelity [\onlinecite{watrous2018theory}, Theorem 3.33]. More examples are given in SM Sec. V. %\cref{sm:sec:other-loss}.

We will say that a recovery scheme is optimal if it has the same threshold as the optimal recovery scheme.

Then, using the two-way bound of the Petz and optimal channels, \cref{eq:petz-opt-two-way-bound}, the threshold for Petz recovery scheme is optimal:
\begin{corollary}[Optimality of the Petz recovery scheme]
    \begin{align}
        p^\mathrm{Petz}_\mathrm{th} = p^\mathrm{opt}_\mathrm{th}.
    \end{align}
\end{corollary}

There are other recovery channels suggested in the literature. Firstly, there is the recovery channel described in Ref. \onlinecite{schumacher2002approximate}, which we will call the Schumacher-Westmoreland (SW) recovery channel, detailed in SM Sec. IV. % \cref{sm:sec:sw-recovery}.
There are also variations of the Petz recovery channel called the rotated Petz recovery channel and its `average', the twirled recovery channel, which admits information-theoretic bounds~\cite{junge2018universal}. It is not known whether these recovery schemes have the same threshold as that of Petz.

Previous works focused on sufficient conditions for error corrections, and there exist quantum information-theoretic \emph{lower bounds} on the optimal entanglement fidelity. The lower bounds are derived specific recovery channels. For example, we have the lower bound for the SW recovery channel \cite{schumacher2002approximate},
\begin{theorem*}[Lower bounds of entanglement fidelity of SW recovery channel, Ref. \onlinecite{schumacher2002approximate}]
    \begin{align} \label{eq:sw-recoverability-bounds}
        F_\mathrm{e}^\mathrm{SW} & \geq F[\rho_{RE}', \rho'_R \otimes \rho'_E] \geq 1 - T[\rho_{RE}', \rho'_R \otimes \rho'_E] \nonumber \\
        & \geq 1 - \sqrt{\frac{\ln 2}{2} I'_{R:E}}.
    \end{align}
\end{theorem*}
Here, $I'_{R:E} = S[\rho'_R] + S[\rho'_E] - S[\rho'_{RE}] = I^\mathrm{c}_{R:Q} - I^{\mathrm{c} \prime}_{R:Q}$ is the mutual information, equal to the decrease in coherent information, where the coherent information is $I^\mathrm{c}_{R:Q}[\sigma_{RQ}] = S[\sigma_{Q}] - S[\sigma_{RQ}]$. $E$ is the environment resulted from dilating the noise channel. $\rho'_{RE}$, $\rho'_R$, $\rho'_E$ are partial traces of $\ket{\Psi_{QRE}'} = V_{QE \leftarrow Q} \ket{\Phi_{RQ}}$, where $V_{QE \leftarrow Q}$ is the dilation of the error channel $\mathcal{E}_Q$. Here, the prime ($\, \prime \,$) denotes the state after applying $\mathcal{E}_Q$. These bounds are useful as they allow for explicit lower bounds of optimal thresholds \cite{pkim2026rigorous}.

However, in order for conditions for information-theoretic quantities that are sufficient \emph{and necessary}, we require an upper bound for $F^\mathrm{opt}_\mathrm{e}$. To the best of the author's knowledge, there are yet no such upper bounds. We now present an information-theoretic upper bound for the entanglement fidelity for the \emph{optimal} recovery channel:
\begin{theorem}[Upper bound for optimal entanglement fidelity] \label{thm:upper-bound-for-optimal-fidelity}
    The entanglement fidelity of the optimal recovery channel is upper bounded as
    \begin{align}
        (F^\mathrm{opt}_\mathrm{e})^2 \leq 1 - \frac{1}{4} T[\rho_{RE}', \rho'_R \otimes \rho'_E]^2.
    \end{align}
    Combining with \cref{eq:sw-recoverability-bounds}, this allows us to construct the following sandwich,
    \begin{align} \label{eq:sandwich}
        1 - T_{R:E}' \leq F^\mathrm{SW}_\mathrm{e} \leq F^\mathrm{opt}_\mathrm{e} \leq \sqrt{1 - \frac{1}{4} (T'_{R:E})^2}.
    \end{align}
\end{theorem}

Here, we defined the mutual trace distance,
\begin{definition}[Mutual trace distance] \label{def:mutual-trace-distance}
    The mutual trace distance between $A$ and $B$ is defined as
    \begin{align}
        T_{A:B}[\rho_{AB}] := T[\rho_{AB}, \rho_A \otimes \rho_B].
    \end{align}
\end{definition}
The proof of \Cref{thm:upper-bound-for-optimal-fidelity} is detailed in SM Sec. I. % \cref{sm:sec:upper-bounds}.
Note that the author is unaware of such two-way bounds via coherent information, whose decrease is given by $I'_{R:E}$. The largest $p$ such that $\lim_{N \rightarrow \infty} I'_{R:E} = 0$ gives a lower bound on $p^\mathrm{opt}_\mathrm{th}$, but when $\lim_{N \rightarrow \infty} I'_{R:E} > 0$, there are no non-vacuous upper bound for $F^\mathrm{opt}_\mathrm{e}$ \emph{in general}. Using \Cref{thm:upper-bound-for-optimal-fidelity} for codes that encode a constant number of logical quantum information, i.e. $d(N) = O(1)$, we can  to construct an analogous sandwich for $I'_{R:E}$:
\begin{corollary}[Two-way bounds for coherent information] \label{cor:zero-rate-coherent-information}
    The decrease in coherent information $I_\mathrm{c} - I'_\mathrm{c} = I'_{R:E}$ two-way bounds the optimal and SW recovery channels as
    \begin{align}
        1 - \sqrt{\frac{\ln 2}{2} I'_{R:E}} \leq F^\mathrm{SW}_\mathrm{e} \leq F^\mathrm{opt}_\mathrm{e} \leq \sqrt{1 - \frac{(I'_{R:E})^3}{4 (2\ln (2d) + 3/e)^3}},
    \end{align}
    where $d$ is the dimension of encoded logical quantum information, e.g. for $K$ encoded qubits, $d = 2^K$.
\end{corollary}
Tighter upper bounds can be made using Corollary 1 of Ref. \onlinecite{shirokov2017tight}, but is not enough for finite-rate codes, as $d(N)$ increases with number of qudits $N$, making the upper bound vacuous. The proof of \Cref{cor:zero-rate-coherent-information} and further discussion is found in SM Sec. II. %\cref{sm:sec:coherent-information-sandwich}.

The first consequence of \Cref{thm:upper-bound-for-optimal-fidelity} is that it allows us to prove that the SW recovery scheme is optimal:
\begin{theorem}[Optimality of the Schumacher-Westmoreland recovery scheme] \label{thm:sw-is-optimal}
    The Schumacher-Westmoreland recovery scheme is optimal, i.e. it has the same threshold as the optimal recovery scheme,
    \begin{align}
        p^\mathrm{opt}_\mathrm{th} = p^\mathrm{SW}_\mathrm{th}.
    \end{align}
\end{theorem}
\begin{proof}
    By \Cref{thm:upper-bound-for-optimal-fidelity}, if we have $\lim_{N \rightarrow \infty} F^\mathrm{opt} = 1$, then $\lim_{N \rightarrow \infty} T'_{R:E} = 0$, and therefore $\lim_{N \rightarrow \infty} F^\mathrm{SW} = 1$. Conversely, if $\lim_{N \rightarrow \infty} F^\mathrm{opt}  < 1$, then we also have $\lim_{N \rightarrow \infty} F^\mathrm{SW} < 1$.
\end{proof}

Using a similar line of reasoning as the proof of \Cref{thm:sw-is-optimal}, the sandwich \cref{eq:sandwich} allows us to prove that
\begin{theorem}[Mutual trace distance is necessary and sufficient] \label{thm:necessary-and-sufficient}
    The mutual trace distance between the register and the environment is necessary and sufficient in determining the optimum threshold,
    \begin{align}
        \begin{aligned}
            \lim_{N \rightarrow \infty} T'_{R:E} & = 0 \Longleftrightarrow \lim_{N \rightarrow \infty} F^\mathrm{opt}_\mathrm{e} = 1, \\
            \lim_{N \rightarrow \infty} T'_{R:E} & > 0 \Longleftrightarrow \lim_{N \rightarrow \infty} F^\mathrm{opt}_\mathrm{e} < 1.
        \end{aligned}
    \end{align}
\end{theorem}
It also allows us to order the optimal thresholds of different error channels, again using the similar line of reasoning as that of \Cref{thm:sw-is-optimal}.
\begin{corollary} \label{cor:thershold-hierachy}
    Consider two error channels $\mathcal{E}_\mathrm{a}$ and $\mathcal{E}_\mathrm{b}$ with the mutual trace distance on $R:E$ bounded as
    \begin{align}
        T^{\mathrm{a}\prime}_{R:E} \leq T^{\mathrm{b}\prime}_{R:E}.
    \end{align}
    Then, the optimal thresholds for the two error channels are bounded as
    \begin{align}
        p^{\mathrm{b},\mathrm{opt}}_\mathrm{th} \leq p^{\mathrm{a},\mathrm{opt}}_\mathrm{th}.
    \end{align}

    In the case that the limits of the mutual trace distances are strictly bounded,
    \begin{align}
        \lim_{N \rightarrow \infty} T^{\mathrm{a}\prime}_{R:E} < \lim_{N \rightarrow \infty}  T^{\mathrm{b}\prime}_{R:E},
    \end{align}
    the optimal thresholds can also be strictly bounded as
    \begin{align}
        p^{\mathrm{b},\mathrm{opt}}_\mathrm{th} < p^{\mathrm{a},\mathrm{opt}}_\mathrm{th}.
    \end{align}
\end{corollary}

For $d = O(1)$, \Cref{cor:zero-rate-coherent-information} allows for similar statements as \Cref{thm:necessary-and-sufficient} and \Cref{cor:thershold-hierachy} for $I'_{R:E}$:
\begin{theorem}
    For codes that encode a constant amount of quantum information, $d(N) = O(1)$, \Cref{thm:necessary-and-sufficient} and \Cref{cor:thershold-hierachy} hold with the replacement of the mutual trace distance with the decrease in coherent information, $T'_{R:E} \rightarrow I'_{R:E} = I_\mathrm{c} - I_\mathrm{c}'$.
\end{theorem}

Let us now prove a few statements that feel intuitively true.

\begin{proposition} \label{lem:meas-syndromes-make-it-worse}
    For a given noise channel $\mathcal{E}$, the optimal fidelity for error correction after application of another channel $\mathcal{M}$ is upper bounded by the optimal fidelity without $\mathcal{M}$,
    \begin{align}
        F_\mathrm{e}^\mathrm{opt}(\mathcal{E}) \geq F_\mathrm{e}^\mathrm{opt}(\mathcal{M} \mathcal{E}).
    \end{align}
\end{proposition}
\begin{proof}
    Consider the optimal recovery channel following measuring the syndromes first, $\mathcal{M}$. Then we have some optimal channel $\mathcal{R}^\mathrm{opt}_{\mathcal{M}\mathcal{E}}$ with fidelity $F_\mathrm{e}(\mathcal{R}^\mathrm{opt}_{\mathcal{M}\mathcal{E}} \mathcal{M} \mathcal{E})$. Now $\tilde{\mathcal{R}} = \mathcal{R}^\mathrm{opt}_\mathrm{meas}\mathcal{M}$ is simply another recovery channel, and therefore lower bounds $F_\mathrm{e}(\mathcal{R}^\mathrm{opt}_{\mathcal{E}} \mathcal{E})$.
\end{proof}
Setting $\mathcal{M}$ as the syndrome measurement channel, this means that the ML decoder threshold, which is a particular recovery scheme applied after measurements of syndromes, is lower bounded by the optimal threshold,
\begin{align}
    p^\mathrm{ML}_\mathrm{th} \leq p^\mathrm{opt}_\mathrm{th}.
\end{align}

\begin{proposition} \label{lem:unitary-dont-change-threshold}
    Given a noise channel $\mathcal{E}$, The optimal fidelity for error correction after applying a unitary $\mathcal{U}$ after the noise channel is the same that without the unitary,
    \begin{align}
        F_\mathrm{e}^\mathrm{opt}(\mathcal{E}) = F_\mathrm{e}^\mathrm{opt}(\mathcal{U} \mathcal{E}).
    \end{align}
\end{proposition}
\begin{proof}
    First, we have $F_\mathrm{e}(\mathcal{R}^\mathrm{opt}_{\mathcal{E}} \mathcal{E}) \geq F_\mathrm{e}(\mathcal{R}^\mathrm{opt}_{\mathcal{U} \mathcal{E}} \mathcal{U} \mathcal{E})$ by the same reasoning used in \Cref{lem:meas-syndromes-make-it-worse}. Now $F_\mathrm{e}(\mathcal{R}^\mathrm{opt}_{\mathcal{E}} \mathcal{E}) = F_\mathrm{e}(\mathcal{R}^\mathrm{opt}_{\mathcal{E}} \mathcal{U}^{-1} \mathcal{U}\mathcal{E})$, and $\mathcal{R}^\mathrm{opt}_{\mathcal{E}} \mathcal{U}^{-1} = \tilde{\mathcal{R}}$ is just another recovery map for $\mathcal{U} \mathcal{E}$, and therefore $F_\mathrm{e}(\mathcal{R}^\mathrm{opt}_{\mathcal{U} \mathcal{E}} \mathcal{U} \mathcal{E}) \geq F_\mathrm{e}(\mathcal{R}^\mathrm{opt}_{\mathcal{E}} \mathcal{E})$.
\end{proof}

\vspace{1.5em}
\paragraph*{Optimal quantum error correction on qubit CSS codes ---}Now that it is established that both the SW and Petz recovery schemes are optimal, we study their structures. To start, we consider Pauli error channels. In this case, it is intuitive that measuring the syndromes then using a Bayes-optimal decoder is the optimal recovery scheme \cite{dennis2002topological,kim2025measurement}. Here, we prove that the Bayes optimal sampling decoder (i.e. one where constructs the posterior then samples from it) is an optimal recovery scheme for qubit CSS codes, such as the 2D surface code [SM Secs. VII A, C] % [SM \cref{sm:sec:pauli-petz,sm:sec:pauli-sw}]
\begin{proposition} \label{thm:pauli-error-is-bayesian-inference}
    For qubit CSS codes, the Petz and SW recovery channel $\mathcal{R}^\mathrm{Petz}_\mathrm{Pauli}$ for Pauli noise $\mathcal{E}_\mathrm{Pauli} [\cdot] = \sum_P p(P) P (\cdot) P^\dagger$, where $P$ is a Pauli string, is equal to a syndrome measurement channel followed by a sampling decoder. Both recovery schemes are optimal.
    \begin{align}
        \mathcal{R}^{\mathrm{Petz}, \mathrm{SW}}_\mathrm{Pauli} = \mathcal{D}^{\mathrm{Petz}, \mathrm{SW}}_\mathrm{sampl.} \mathcal{M}_\mathrm{syndr.}
    \end{align}
    For the Petz recovery, the sampling decoder samples from the Bayes-optimal posterior distribution
    \begin{align}
        p(\ell \vert s) = \frac{p(\ell, s)}{\sum_{\ell'} p(\ell', s)}.
    \end{align}
    This is in contrast to the SW recovery, which samples from
    \begin{align}
        q(\ell \vert s) = \frac{p(\ell, s)^2}{\sum_{\ell'} p(\ell', s)^2}.
    \end{align}
    Here, $p(s, \ell)$ is the joint probability of syndromes $s$ and logical error $\ell$.
\end{proposition}
The SW recovery scheme has an effect of picking the more likely logical sector more often. Nevertheless, it is optimal and shares the threshold with the Petz recovery scheme. Petz and SW recovery channels for Pauli errors consists of measuring syndromes then using a classifier (decoder) to match the errors. This equates $F_\mathrm{e}$ to the average classification performance of the decoder. The Bayes-optimal classifier theorem \cite{bishop2006pattern,kim2025measurement}, which states that the best classifer is one that takes the $\mathrm{argmax}$ of the Bayes-optimal posterior (i.e. ML classifier), can then be used to show
\begin{corollary}[ML decoder is optimal for Pauli errors]
    For Pauli errors on qubit CSS codes, the maximum likelihood recovery channel $\mathcal{R}^\mathrm{ML} = \mathcal{D}_\mathrm{ML} \mathcal{M}_\mathrm{syndr.}$ performs no worse than Petz and SW recovery and therefore is an optimal recovery scheme,
    \begin{align}
        F_\mathrm{e} (\mathcal{R}^\mathrm{Petz,SW}) \leq F_\mathrm{e} (\mathcal{R}^\mathrm{ML}) \implies p_\mathrm{th}^\mathrm{ML} = p_\mathrm{th}^\mathrm{opt}.
    \end{align}
\end{corollary}
The ML classifier could be framed probabilistically as $p_\mathrm{ML}(\ell, s) = \delta_{\mathrm{argmax}_{\ell'} p(\ell' \vert s), \ell}$, or $\alpha \rightarrow \infty$ for $q_\alpha(\ell \vert s) = p(\ell, s)^\alpha / \sum_{\ell'} p(\ell', s)^\alpha$.

Combining \Cref{thm:pauli-error-is-bayesian-inference} with \Cref{lem:unitary-dont-change-threshold}, we can now prove the intuitive statement of the introduction [SM Sec. VI]: % \cref{sm:sec:coherent-rotations}]:
\begin{proposition}
    Assume that the maximum-likelihood decoder threshold for a qubit CSS code affected by bit-flip errors exists and is equal to $p_\mathrm{c}$. Then the separable noise-channel $\mathcal{E} = \prod_i \mathcal{E}_i$ with
    \begin{align}
        \mathcal{E}_i[\cdot] = \int d \theta p(\theta) U_\theta (\cdot) U^\dag_\theta, \quad U_\theta = e^{i\theta X},
    \end{align}
    can be written as $\mathcal{U}_\phi \mathcal{E}_\mathrm{bit-flip}$,
    where $\mathcal{U}_\phi = \prod_i \mathcal{U}_i^\phi$ is a single-site unitary channel and $\mathcal{E}_\mathrm{bit-flip}$ is a bit-flip error channel with $\tilde p = (1-\sqrt{\Delta^2 + 4t^2})/2$, $\Delta = \langle \cos^2 \theta - \sin^2 \theta \rangle$ and $t = \langle \cos \theta \sin \theta \rangle$. An optimal recovery scheme is $\mathcal{D}_\mathrm{ML} \mathcal{M}_\mathrm{syndr.} \mathcal{U}^{-1}_\phi$, with the optimal threshold $\tilde p = p_\mathrm{c} = p^\mathrm{opt}_\mathrm{th}$.
\end{proposition}

A more interesting case is the recovery for the channel consisting of single-site amplitude damping. Its dilated channel on a single qubit can be written as
\begin{align}
    & V_{QE \leftarrow Q}^\mathrm{ad} \ket{\psi}_Q \nonumber \\
    & = \left(\frac{1 + \sqrt{1-p}}{2} I + \frac{1-\sqrt{1-p}}{2} Z \right) \ket{\psi}_Q \otimes \ket{0}_E \nonumber \\
    & + \frac{\sqrt{p}}{2} (X + ZX) \ket{\psi}_Q \otimes \ket{1}_E.
\end{align}
This can be interpreted in the following way. The $X$ error is incoherent (once tracing over the environment $E$). However, when the bit-flip error occurs, it also creates a coherent superposition of $Z$ errors. Hence, one may expect that the optimal recovery scheme is to measure the $Z$-syndromes (which finds $X$-type errors) but perform a coherent operation for $X$-errors. This is what we find upon explicit construction of optimal recovery schemes [SM Secs VII D, E]: % \cref{sm:sec:ad-petz,sm:sec:ad-sw}]:
\begin{proposition}
    For the ampltiude damping channel and a qubit CSS code, both Petz and SW recovery channels consist of measuring the $Z$-type syndromes,
    \begin{align}
        \mathcal{R}^\mathrm{Petz, SW}_\mathrm{AD} = \tilde{\mathcal{R}}^\mathrm{Petz, SW}_\mathrm{AD} \mathcal{M}_{Z\mathrm{\, syndr.}}.
    \end{align}
\end{proposition}
Determining the exact form of $\tilde{\mathcal{R}}^\mathrm{Petz, SW}_\mathrm{AD}$ is not tractable as one must diagonalise the matrix $\mathcal{E}(\mathbb{\Pi})$ to compute $\mathcal{E}(\mathbb{\Pi})^{-1/2}$, where $\mathbb{\Pi}$ is the code-space projector. However, for Petz we show that the recovery channel can be written in the form $\mathcal{R}^\mathrm{Petz}_\mathrm{AD}[\cdot] = \sum_{\ell_z s_z b} R_{\ell_z s_z b} (\cdot) R^\dagger_{\ell_z s_z b}$, where the Kraus operator is
\begin{align}
    R_{\ell_z s_z b} = \bar{X}^{\ell_z} X^{s_z} \left[\sum_{\ell_x, s_x} {\hat w}^{\ell_z s_z b}_{\ell_x s_x} Z^{s_x} \bar{Z}^{\ell_x} \mathbb{\Pi}_{s_x}\right] \mathbb{\Pi}_{s_z}.
\end{align}
Here, $\ell_x$ and $s_x$ indexes $X$ logical error and syndromes, respectively, and similar for $\ell_z, s_z$. On the right, we have $\mathbb{\Pi}_{s_z}$ which is the projective measurement onto the $Z$-syndromes. Due to the summation over $\ell_x, s_x$ within a single Kraus operator, the term in square brackets is \emph{not} a simple measurement of $s_x$-syndromes but rather a coherent recovery operation, that superposes $X$-syndrome and logical corrections weighted and phased appropriately by ${\hat w}^{\ell_z s_z b}_{\ell_x s_x}$ with outcomes $b$. In theory, a polar decomposition of the bracketed term into a form $p_{\ell_z s_z b} \, U^b \, \mathbb{\Pi}_b$ would reveal the basis of measurement and recovery operation.

\vspace{1.5em}
\paragraph*{Optimal and non-optimal quantum recovery problems ---}Just as the phase diagram of classical or quantum Bayesian inference or error-correction problems (with the recovery scheme of matching after measurement of syndromes) can be expanded to non-optimality \cite{p2025planted,kim2025measurement,zdeborova2016statistical}, now the same can be done for quantum recovery problems. We will choose the Petz recovery scheme as a concrete example, since it is the quantum analogue of classical Bayesian inference [\onlinecite{petz1986sufficient}, \onlinecite{bai2025quantum}].
Then non-optimal quantum recovery problems can be perturbed from optimal quantum recovery problems in the following way. Consider a `teacher' who applies a noise channel $\mathcal{E}_*^Q$ with parameter $p_*$ on an entangled state $\ket{\Psi_{RQ}}$. The `student' does not know which error channel the teacher used and constructs an optimal recovery scheme assuming the noise channel $\mathcal{E}_\mathrm{s}^Q$. Let us specialise to the case where the student knows functional form of the initial state and the noise channel, but may assume the wrong parameters $p_\mathrm{s}$. Under this setting, the recovery phase diagram can be expanded in $(1/p_\mathrm{s}, 1/p_*)$ parameter space from the optimal setting $p_* = p_\mathrm{s}$ or more generally $\mathcal{E}^Q_* = \mathcal{E}_\mathrm{s}^Q$ if we allow for more than one parameter in the error channel. In fact, for Pauli errors on a qubit CSS code, using \Cref{thm:pauli-error-is-bayesian-inference}, this exactly maps to a quantum inference problem \cite{kim2025measurement,dennis2002topological}.

Just as in classical or quantum Bayesian inference problems, exploiting the optimality of the Petz recovery scheme, we can make the following statement about the shape of the phase diagram:
\begin{proposition}[Shape of phase diagram of non-optimal quantum recovery problems] \label{lem:non-optimal-recovery}
    If a given optimal point $(1/p_\mathrm{s} = 1/p_*, 1/p_*)$ is in the perfect recovery impossible phase, then so are all points on the horizontal line passing through it. A corollary is that the phase boundary cannot pass through values of $1/p_*$ below the threshold $1/p^\mathrm{opt}_\mathrm{th}$ on the Bayes optimal line. This is illustrated in \cref{fig:non-optimal-phase-diagram}.
\end{proposition}

\begin{figure}
    \centering
    \includegraphics[width=0.7\linewidth]{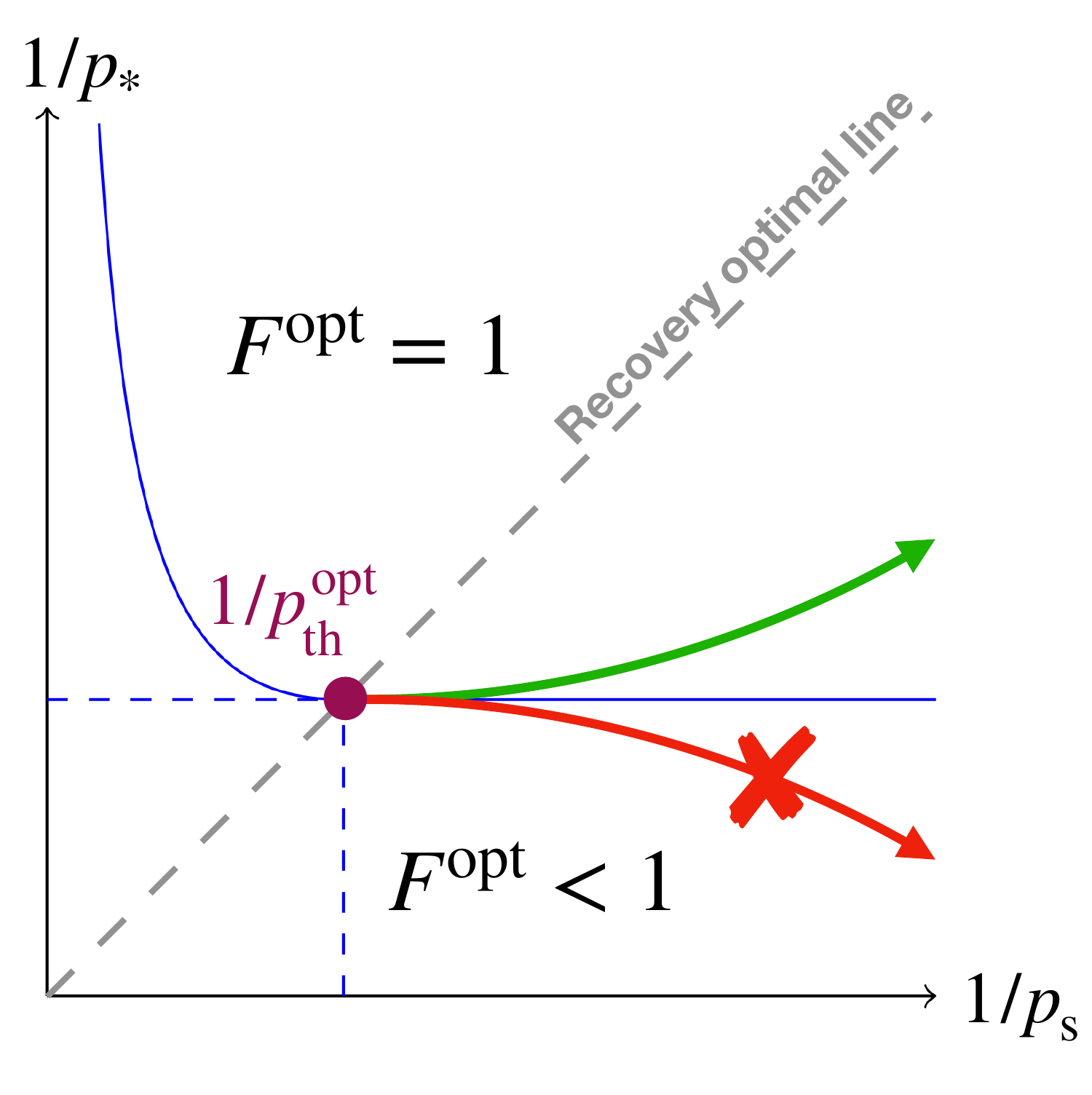}
    \caption{Illustration of \Cref{lem:non-optimal-recovery} for a schematic recoverability phase diagram. Since Petz recovery maps with $p_* = p_\mathrm{s}$ is the optimal, any parameter on the horizontal of the recovery optimal line must recover as equally or worse.}
    \label{fig:non-optimal-phase-diagram}
\end{figure}

\vspace{1.5em}
\paragraph*{Discussion and Outlook ---}
% \section{Outlook}
There are a number of future research directions. First is to use the newly introduced mutual trace distance $T'_{R:E}$ to determine absolutely optimal thresholds for codes without a constant amount of encoded quantum information, e.g. finite-rate codes. Such programme has been executed for coherent information for zero-rate codes \cite{colmenarez2024accurate}. Another direction would be to further explore the structure of the two recovery schemes proved to be optimal here, namely the Petz and Schumacher-Westmoreland recovery schemes, and see if there are any methods to implement some of their features in a practical recovery scheme. It would also be interesting to know if other known recovery schemes such as the rotated Petz, twirled Petz are also optimal or not~\cite{junge2018universal}. Also would be to know if one can construct a rigorous bound between the optimal thresholds of syndrome measurement after errors $\mathcal{M} \mathcal{E}$ and that of Pauli-twirled error channels $\tilde{\mathcal{E}}$, which is numerically observed in Ref. \onlinecite{behrends2025statistical}. 

Note that in the present work, we put no restrictions on the structure of $\mathcal{R}$. Depending on the purpose, one may add additional restrictions to the recovery scheme:
\begin{enumerate}
    \item[(A)] If there is no restriction on the recovery channel, then non-local quantum as well as expensive classical processing is allowed.
    \item[(B)] For practical purposes, one may only allow for local quantum operations, as non-local (e.g. high Pauli weight operators for stabiliser codes) measurements and unitaries are often unfeasible. Meanwhile, classical processing can still be non-local.
    \item[(C)] The recovery channel could be forced to follow stabiliser measurements, while allowing non-local classical processing. This is equivalent to replacing the noise channel $\mathcal{E}$ with $\mathcal{M}\mathcal{E}$, where $\mathcal{M}$ is the syndrome measurement channel.
    \item[(D)] The classical processing (i.e. the `decoder') could be limited to those that require memory or timescales that scale polynomially with the number of physical qudits $N$. For example, for the toric code, the timescale for minimal-weight perfect matching decoder is $\mathrm{poly}(N)$ while that of the maximum-likelihood decoder is $\mathrm{exp}(N)$ \cite{bravyi2014efficient}.
    \item[(E)] For physicists, the recovery channel could be restricted to those where both quantum and classical processing is constant-depth local. This coincides with the typical definition of the noisy state being in the same open phase of matter as the clean state, since the two states would be related by local channels\footnote{Here, the classical resources would not necessarily be restricted. We have also assumed that the noise channel is local.} \cite{sang2024mixed}
    $$\mathcal{E} \rho \mathop{\longrightarrow}_{\mathcal{R}_\mathrm{loc}} \rho, \quad \rho \mathop{\longrightarrow}_{\mathcal{E}} \mathcal{E} \rho.$$
\end{enumerate}

The optimal recovery scheme $\mathcal{R}^\mathrm{opt}$ its threshold $p^\mathrm{opt}_\mathrm{th}$ studied in this work is within the context of (A). This should be an upper bound for any best thresholds within (B-E). It would be interesting to explore formal statements about optimal recovery schemes under some of the above restrictions. A formal statement within (E) in similar fashion to the current work, combined with existing results for sufficiency \cite{sang2024mixed,sang2025stability} may provide a necessary and sufficient information-theoretic quantity to diagnose mixed-state quantum phases.

\vspace{1.5em}
\paragraph*{Acknowledgments ---}
The author is supported by UKRI Engineering and Physical Sciences Research Council (EPSRC) DTP International Studentship Grant Ref.\ No.\ EP/W524475/1. I would like to thank Max McGinley for drawing my attention to Ref. \onlinecite{shirokov2017tight}, Mircea Bejan for pointing out an algebraic error, Curt von Keyserlingk and Austen Lamacraft for reading and giving comments to earlier drafts. I thank all above for useful discussions.

\bibliography{bibliography}



\section*{Supplementary material}

# Supplemental Material for "Optimal recovery for quantum error correction"

Sun Woo P. Kim*

*Department of Physics, King's College London, Strand, London WC2R 2LS, United Kingdom*

**CONTENTS**

* swk34@cantab.ac.uk

# I. PROOF OF THEOREM 1: UPPER BOUNDS ON OPTIMAL ENTANGLEMENT FIDELITY

By purification of quantum channels and Uhlmann's theorem,

$$F_{\text{e}}^{\text{opt}} = \max_{\mathcal{R}} F[\mathcal{R}_Q \rho'_{QR}, \Phi_{QR}] = \max_{\mathcal{R}} F[\mathcal{R}_Q \mathcal{E}_Q \Phi_{QR}, \Phi_{QR}] \tag{S1}$$

$$= \max_{W_{EE'}, V_{QE'}} F\left[V_{QE'} U_{QE} \ket{\Phi}_{QR} \ket{0}_E \ket{0}_{E'}, W_{EE'} \ket{\Phi_{QR}} \ket{0}_E \ket{0}_{E'}\right]. \tag{S2}$$

Here, we purified the system into $E$ and $E'$. $V_{QE'}$ is the purification of recovery map $\mathcal{R}_Q$, $U_{QE}$ is the purification of $\mathcal{E}_Q$, and $W_{EE'}$ is some optimal unitary to be equal to $F_{\text{e}}^{\text{opt}}$. We have chosen to purify $\ket{\Phi}_{QR}$ in the same way as $\mathcal{R}_Q \rho'_{QR}$; this does not matter as $W_{EE'}$ will be chosen to maximise the fidelity. Now we trace over $Q$ and $E'$. This removes the recovery map as we only look at the subsystem $(QE')^c = RE$. Denote some channel $\mathcal{W}_E[\cdot] = \text{tr}_{E'}[W_{EE'} \cdot W_{EE'}^\dagger]$. Then we have

$$F_{\text{e}}^{\text{opt}} = \max_{\mathcal{W}_E} F\left[\mathcal{E}_{E \leftarrow Q}^c[\Phi_{QR}], \rho'_R \otimes \mathcal{W}_E[\ket{0}\bra{0}_E]\right] = \max_{\sigma_E} F\left[\rho'_{RE}, \rho'_R \otimes \sigma_E\right]. \tag{S3}$$

Here, we used the fact that $\rho'_R = \rho_R$, since no channel is applied on $R$. Now fix the optimal $\sigma_E^* = \text{argmax}_{\sigma_E} F[\rho'_{RE}, \rho'_R \otimes \sigma_E]$. Denote

$$F_{\text{e}}^{\text{opt}} = F[\rho'_{RE}, \rho'_R \otimes \sigma_E^*] = F_{R:E}. \tag{S4}$$

By data processing inequality [S1], we have

$$F_{R:E} = F[\rho'_{RE}, \rho'_R \otimes \sigma_E^*] \leq F[\rho'_E, \sigma_E^*] =: F_E. \tag{S5}$$

Then we have

$$F_{R:E}^2 \leq F_E^2 \implies 1 - F_E^2 \leq 1 - F_{R:E}^2 \implies \sqrt{1 - F_E^2} \leq \sqrt{1 - F_{R:E}^2}. \tag{S6}$$

Now look at

$$\|\rho'_{RE} - \rho'_R \otimes \rho'_E\|_1 = \|(\rho'_{RE} - \rho'_R \otimes \sigma_E^*) + (\rho'_R \otimes \sigma_E^* - \rho'_R \otimes \rho'_E)\|_1 \tag{S7}$$

$$\leq \|\rho'_{RE} - \rho'_R \otimes \sigma_E^*\|_1 + \|\rho'_R \otimes \sigma_E^* - \rho'_R \otimes \rho'_E\|_1. \tag{S8}$$

Note that $\|\rho'_R \otimes \sigma_E^* - \rho'_R \otimes \rho'_E\|_1 = \|\rho'_R - \rho'_R\|_1 + \|\sigma_E^* - \rho'_E\|_1 = 0 + \|\sigma_E^* - \rho'_E\|_1$. Denote $\tilde{T}_{R:E} = \frac{1}{2}\|\rho'_{RE} - \rho'_R \otimes \rho'_E\|_1$, $T_E = \frac{1}{2}\|\sigma_E^* - \rho'_E\|_1$, and $T_{R:E} = \|\rho'_{RE} - \rho'_R \otimes \sigma_E^*\|_1$. Then we have

$$\tilde{T}_{R:E} \leq T_{R:E} + T_E. \tag{S9}$$

Now using the Fuchs-van de Graaf inequalities [S2, Theorem 3.33], we have $T_E \leq \sqrt{1 - F_E^2} \leq \sqrt{1 - F_{R:E}^2}$, so we have

$$\tilde{T}_{R:E} \leq T_{R:E} + \sqrt{1 - F_{R:E}^2}. \tag{S10}$$

We also have $T_{R:E} \leq \sqrt{1 - F_{R:E}^2}$, so we have

$$\tilde{T}_{R:E} \leq 2\sqrt{1 - F_{R:E}^2}. \tag{S11}$$

Rearranging, we have

$$(F_{\text{e}}^{\text{opt}})^2 = F_{R:E}^2 \leq 1 - \frac{1}{4}\tilde{T}_{R:E}^2, \tag{S12}$$

as required.

## II. PROOF OF COROLLARY 2: TWO-WAY BOUNDS FOR COHERENT INFORMATION

From Corollary 1 of Ref. S3, set $A = R$, $B = E$, and $C$ to be the trivial Hilbert space. let $\rho = \rho'_{RE}$ and $\sigma = \rho'_R \otimes \rho'_E$. Since $I_{R:E}[\rho'_R \otimes \rho'_E] = 0$, this gives the upper bound

$$I'_{R:E}[\rho'_{R:E}] \leq 2T'_{R:E} \ln(\min\{d, d_E\}) + (1 + T'_{R:E}) \ln(1 + T'_{R:E}) - T'_{R:E} \ln(T'_{R:E}). \tag{S13}$$

Now we have the upper bound $(1 + T) \ln(1 + T) \leq 2\ln(2)T \forall T \in [0, 1]$. We also have the upper bound $-T'_{R:E} \ln(T'_{R:E}) \leq \frac{1}{e(1-\alpha)} T^\alpha \forall \alpha < 1, T \in [0, 1]$. We assume that $d \leq d_E$; this is a sensible assumption since, for example, error channels acting on single physical qudits that have at least as many Kraus operators as the dimension of the physical qudits, $d_E \geq d_{\text{physical}}^N$, where $d_{\text{physical}}$ is the dimension of individual physical qudit degrees of freedom, and we strictly have $d < d_{\text{physical}}^N$ (i.e. one cannot encode more information than that of the physical qudits.). So we have

$$I'_{R:E} \leq 2\ln(2d) T'_{R:E} + \frac{1}{e(1-\alpha)} (T'_{R:E})^\alpha. \tag{S14}$$

For cleanliness we choose $\alpha = 2/3$. Then, using the fact that $T \leq T^\alpha \forall \alpha < 1, T \in [0, 1]$, we retrieve the upper bounded stated in the main text.

Let us now consider when $d(N) > O(1)$. The mutual information is upper bounded as $I'_{R:E} \leq 2 \ln(d(N))$. Note that for the bound Equation (S13) to be a useful bound, It is not enough for it to tend to a constant value, but in fact to grow at least as fast as $I'_{R:E} \propto \ln d(N)$ in order to provide a non-vacuous upper bound for the optimal entanglement fidelity $F^{\text{opt}}$.

## III. PETZ RECOVERY CHANNEL

The Petz recovery channel to recover state $\sigma$ is defined as

$$\mathcal{R}^{\text{Petz}}_\sigma[\cdot] := \sigma^{1/2} \mathcal{E}^\dagger[\mathcal{E}(\sigma)^{-1/2}(\cdot)\mathcal{E}(\sigma)^{-1/2}]\sigma^{1/2}. \tag{S15}$$

For our purposes (and all bounds on thresholds, etc. stated in the main text), we set $\sigma = \Pi/d$, where $\Pi$ is the code space projector, obtaining

$$\mathcal{R}^{\text{Petz}}[\cdot] := \Pi \mathcal{E}^\dagger[\mathcal{E}(\Pi)^{-1/2}(\cdot)\mathcal{E}(\Pi)^{-1/2}]\Pi. \tag{S16}$$

In order for $\mathcal{E}(\Pi)$ to have a proper inverse, it must have support on the entire Hilbert space of $Q$. In the case that it does not (for example, consider a bit-flip error on the toric code, which does not create any $X$-type syndromes), we define $\mathcal{E}(\Pi)^{-1/2}$ as the square root of the inverse defined on the support of $\mathcal{E}(\Pi)$ and arbitrarily completed on the rest of the Hilbert space.

## IV. SCHUMACHER-WESTMORELAND RECOVERY CHANNEL

The Schumacher-Westmoreland recovery channel is defined as follows [S1]. First, an entangled state between the logical qudit $Q$ with code space dimension $d$ and register $R$ is fixed,

$$|\Phi_{QR}\rangle = \sum_{k=0}^{d-1} \sqrt{\lambda_k} |\overline{k}\rangle_Q |k\rangle_R. \tag{S17}$$

Then we consider the joint purified state after the application of the noise channel $\mathcal{E}_Q$,

$$|\Psi'_{QRE}\rangle = \frac{1}{\sqrt{d}} \sum_k U_{QE} |\overline{k}\rangle_Q |k\rangle_R |\chi\rangle_E, \tag{S18}$$

such that $\text{tr}_E[U_{QE}(\cdot \otimes |\chi\rangle\langle\chi|_E)U_{QE}^\dagger] = \mathcal{E}[\cdot]$.

Since there are no operations applied on $R$, we have

$$\rho'_R = \rho_R = \sum_k \lambda_k |k\rangle\langle k|_R, \quad \rho'_E = \sum_k \frac{1}{d} \mathcal{E}^c_{E \leftarrow Q} \left[ |\bar{k}\rangle\langle \bar{k}|_Q \right]. \tag{S19}$$

Now we purify $\rho'_R \otimes \rho'_E$. The naive purification leads to a pure state in $Q'RE$, where $Q' = R'R''Q$, although often times, purification to a smaller space is sufficient. Then we have the purification

$$|\Omega_{Q'RE}\rangle = \sum_k \sqrt{\lambda_k} |k\rangle_R |k\rangle_{R''} \sum_{k'} U_{QE} |k'\rangle_{R'} |\bar{k'}\rangle_Q |\chi\rangle_E. \tag{S20}$$

Now consider any extension of $|\Psi'_{QRE}\rangle$ into $Q'RE$, $|\Psi'_{Q'RE}\rangle := |\Psi'_{Q'RE}\rangle |\eta\rangle_{R'R''}$. We consider the fidelity

$$F[\rho'_{RE}, \rho'_R \otimes \rho'_E] = \max_{V_{Q'}} |\langle \Psi'_{Q'RE}|V_{Q'}|\Omega_{Q'RE}\rangle| = \max_{V_{Q'}} |\text{tr}\left[ V_{Q'} |\Omega_{Q'RE}\rangle\langle \Psi'_{Q'RE}| \right]| \tag{S21}$$

$$= \max_{V_{Q'}} |\text{tr}_{Q'} [V_{Q'} M_{Q'}]|, \tag{S22}$$

where

$$M_{Q'} = \text{tr}_{RE} \left[ |\Omega_{Q'RE}\rangle\langle \Psi'_{Q'RE}| \right]. \tag{S23}$$

Now, the optimal $V_{Q'}^{\text{argmax}} := \text{argmax}_{V_{Q'}} |\text{tr}_{Q'} [V_{Q'} M_{Q'}]|$ can be determined in the following. Consider the SVD decomposition of $M_{Q'} = U_{Q'} \Lambda_{Q'} W_{Q'}^\dagger$. Since $\Lambda_{Q'}$ is a positive definite matrix, the choice $V_{Q'}^{\text{argmax}} = W_{Q'} U_{Q'}^\dagger$ maximises the trace. Then, we write $V_{Q'}^{\text{argmax}} |\Omega_{Q'RE}\rangle$ in the form

$$V_{Q'}^{\text{argmax}} |\Omega_{Q'RE}\rangle = \sum_{kl} \sqrt{\mu_k \nu_l} |k\rangle_R |\phi^{kl}\rangle_Q |l\rangle_E, \tag{S24}$$

where $\{|\phi^{kl}\rangle_Q\}_{kl}$ forms an orthonormal basis. This is guaranteed by the separability of $\rho'_R \otimes \rho'_E$.

Then the recovrey scheme is as follows. We first make a POVM

$$\Pi_l = \sum_a |\phi^{kl}\rangle\langle \phi^{kl}|, \tag{S25}$$

followed by a unitary that maps $|\phi^{kl}\rangle_Q$ back to the original state $|\bar{k}\rangle_Q$.

## V. TWO-WAY BOUNDS AND OPTIMAL THRESHOLDS FOR OTHER LOSS FUNCTIONS

We will denote the 'average fidelity' defined in Ref. S4 as the Root-Mean-Squared Haar (RMSH) fidelity, where

$$F_{\text{RMSH}}(\mathcal{E}) := \sqrt{\int d\mu_{\text{Haar}}(\psi) F\left[\mathcal{E}(\psi), \psi\right]^2}, \tag{S26}$$

where $\psi = |\psi\rangle\langle\psi|$ is over pure states in code space. It is proven in Ref. S5 that

$$F_{\text{RMSH}}(\mathcal{E}) = \sqrt{\frac{F_e(\mathcal{E})^2 + 1/d}{1 + 1/d}}. \tag{S27}$$

Let the worst-case or minimum pure state fidelity, be defined as

$$F_{\text{wp}}(\mathcal{E}) = \min_\psi F[\mathcal{E}\psi, \psi], \tag{S28}$$

where the minimum is taken over pure states in the code space. As the minimum is always less than the mean, we have $F_{\text{wp}}(\mathcal{E}) \leq F_{\text{RMSH}}(\mathcal{E})$.

One can also define the worst-case or minimum mixed state fidelity as [S6]

$$F_{\text{wm}}(\mathcal{E}) = \min_{\rho} F[\mathcal{E}\rho, \rho], \tag{S29}$$

where the minimum is taken over mixed states in the code space. By using the joint concavity of fidelity [S2, Theorem 3.25], one can show that $F_{\text{wp}}(\mathcal{E}) = F_{\text{wm}}(\mathcal{E})$.

Let $\rho_A$ be an arbitrary density matrix in a Hilbert space whose dimension equal to that of the code space. Then for any such $\rho_A$, we have

$$T[\mathcal{R}_Q \mathcal{E}_Q \mathcal{C}_{Q \leftarrow A} \rho_A, \mathcal{C}_{Q \leftarrow A} \rho_A] \leq \frac{1}{2} \|\mathcal{R}_Q \mathcal{E}_Q \mathcal{C}_{Q \leftarrow A} - \mathcal{C}_{Q \leftarrow A}\|_{\diamond} \|\rho_A\|_1, \tag{S30}$$

where we have interpreted $\mathcal{C}_{Q \leftarrow A}$ as a code preparation / encoding channel. $\|\cdot\|_{\diamond}$ is the diamond norm, also known as the completely bounded trace norm [S2, Definition 3.43]. Now we use the Choi-representation and its equality to the diamond norm [S2, Eq. (3.414)], where for any linear map on density matrices $\mathcal{M}_{Q \leftarrow A}$,

$$\|\mathcal{M}_{Q \leftarrow A}\|_{\diamond} = \|J(\mathcal{M}_{Q \leftarrow A})\|_1 := d_A \|\mathcal{M}_{Q \leftarrow A} \Phi_{AA'}\|_1, \tag{S31}$$

where $\Phi_{AA'}$ is a Bell state between $A$ and an auxillary space $A'$ with the same dimension. Applying to our case and setting $\rho_A$ as the worst case mixed density matrix, we have

$$1 - F_{\text{wm}}(\mathcal{E}) \leq d\sqrt{1 - F_e(\mathcal{E})^2} \implies 1 - d\sqrt{1 - F_e(\mathcal{E})^2} \leq F_{\text{wm}}(\mathcal{E}). \tag{S32}$$

Therefore we have the sandwich

$$1 - d\sqrt{1 - F_e(\mathcal{E})^2} \leq F_{\text{wp}}(\mathcal{E}) = F_{\text{wm}}(\mathcal{E}) \leq F_{\text{RMSH}}(\mathcal{E}) = \sqrt{\frac{F_e(\mathcal{E})^2 + 1/d}{1 + 1/d}}. \tag{S33}$$

The equality for $F_{\text{HMRS}}$ immediately gives that $p_{\text{th}}^{\text{opt}}(F_{\text{RMSH}}) = p_{\text{th}}^{\text{opt}}(F_e)$. For worst-case fidelity, in the $d(N) = O(1)$ case, we can show that $p_{\text{th}}^{\text{opt}}(F_{\text{wm,p}}) = p_{\text{th}}^{\text{opt}}(F_e)$ in the following.

Consider the $N \to \infty$ limit for the $d(N) = O(1)$ case:

- Let $\mathcal{E} = \mathcal{R}_Q^{\text{opt(e)}} \mathcal{E}_Q$, where $\mathcal{R}_Q^{\text{opt(e)}}$ is the optimal recovery channel for the loss function $\mathscr{L} = 1 - F_e$. This defines a possibly non-optimal recovery channel for the other loss functions. If $F_e^{\text{opt(e)}} = 1$, then we must have $F_{\text{wm,p}}^{\text{opt(e)}} = 1$ due to the lower bound. This means that we must have $p_{\text{th}}^{\text{opt}}(F_{\text{wm,p}}) \geq p_{\text{th}}^{\text{opt}}(F_e)$, since wherever $F_e^{\text{opt(e)}} = 1$, so are some (possibly suboptimal) recovery maps for worst-case fidelities.

- Conversely, choose the optimal recovery channel with respect to the worst-case fidelity. If $F_{\text{wm}}^{\text{opt(wm,p)}} = 1$, then we must have $F_e^{\text{opt(wm,p)}} = F_e^{\text{opt(wm,p)}} = 1$ due to the upper bound. Therefore we have $p_{\text{th}}^{\text{opt}}(F_{\text{wm,p}}) \leq p_{\text{th}}^{\text{opt}}(F_e)$.

- Together, we have $p_{\text{th}}^{\text{opt}}(F_{\text{wm,p}}) = p_{\text{th}}^{\text{opt}}(F_e)$.

In the $d(N) > O(1)$ case, we can show that $p_{\text{th}}^{\text{opt}}(F_{\text{wm,p}}) \geq p_{\text{th}}^{\text{opt}}(F_e)$ but must assume convergence properties if they are to be equal:

- Pick any of the optimal recovery channel with respect to worst-case fidelity. When $F_{\text{wm,p}}^{\text{opt(wm,p)}} = 1$, then we must have $F_e^{\text{opt(wm,p)}} = 1$ due to the upper bound. Therefore we have $p_{\text{th}}^{\text{opt}}(F_{\text{wm,p}}) \geq p_{\text{th}}^{\text{opt}}(F_e)$.

- However, if we now pick an optimal recovery channel with respect to e, even if $F_e^{\text{opt(e)}} \to 1$, we cannot guarantee a non-vacuous lower bound for $F_{\text{wm,p}}^{\text{opt(e)}}$, unless we assume that $1 - F_e^{\text{opt(e)}}$ decays faster than $1/d^2$.

The author is unaware if there exist better upper and lower bounds for the different fidelity metrics than the ones stated here. A sharper lower bound would be helpful for the $d(N) > O(1)$ case.

## VI. COHERENT ROTATION CHANNELS WITH UNCERTAINTY AND PROOF OF PROPOSITION 4

Consider the single-qubit noise channel

$$\mathcal{E}[\cdot] = \int da d\vec{n} p(a, \vec{n}) e^{ia\vec{n}\cdot\vec{\sigma}} \cdot e^{-ia\vec{n}\cdot\vec{\sigma}}. \tag{S34}$$

Now using the identity $e^{ia\vec{n}\cdot\vec{\sigma}} = \sum_{\mu=0}^{3} c_\mu \sigma_\mu$, where $c_0 = \cos(a)$, $c_i = i\sin(a) n_i$, we have

$$\mathcal{E}[\cdot] = \int da d\vec{n} p(a, \vec{n}) \sum_{\mu\nu} c_\mu c_\nu^* \sigma_\mu \cdot \sigma_\nu = \sum_{\mu\nu} \langle c_\mu c_\nu^* \rangle \sigma_\mu \cdot \sigma_\nu^\dagger. \tag{S35}$$

The matrix $M_{\mu\nu} := \langle c_\mu c_\nu^* \rangle$ is a Hermitian matrix given by

$$M_{\mu\nu} = \begin{pmatrix} \langle \cos^2 a \rangle & -i\langle \cos a \sin a n_x \rangle & -i\langle \cos a \sin a n_y \rangle & -i\langle \cos a \sin a n_z \rangle \\ i\langle \cos a \sin a n_x \rangle & \langle \sin^2 a n_x^2 \rangle & \langle \sin^2 a n_x n_y \rangle & \langle \sin^2 a n_x n_z \rangle \\ i\langle \cos a \sin a n_y \rangle & \langle \sin^2 a n_y n_x \rangle & \langle \sin^2 a n_y^2 \rangle & \langle \sin^2 a n_y n_z \rangle \\ i\langle \cos a \sin a n_z \rangle & \langle \sin^2 a n_z n_x \rangle & \langle \sin^2 a n_z n_y \rangle & \langle \sin^2 a n_z^2 \rangle \end{pmatrix}. \tag{S36}$$

Diagonalising the matrix as $O\Lambda O^\dagger$, we have

$$\mathcal{E}[\cdot] = \sum_{\mu\nu\alpha} O_{\mu,\alpha} \lambda_\alpha O^\dagger_{\alpha,\nu} \sigma_\mu \sigma_\nu^\dagger \tag{S37}$$

$$= \sum_\alpha \lambda_\alpha \left( \sum_\mu O_{\mu,\alpha} \sigma_\mu \right) \cdot \left( \sum_\nu O^\dagger_{\alpha,\nu} \sigma_\nu^\dagger \right) \tag{S38}$$

$$= \sum_{\alpha=0}^{3} \lambda_\alpha K_\alpha \cdot K_\alpha^\dagger. \tag{S39}$$

Consider the case $n_y = n_z = 0$, i.e. rotations only in the $X$-axis. In this case, the matrix simplifies to

$$M_{\mu\nu} = \begin{pmatrix} \langle \cos^2 a \rangle & -i\langle \cos a \sin a \rangle & 0 & 0 \\ i\langle \cos a \sin a \rangle & \langle \sin^2 a \rangle & 0 & 0 \\ 0 & 0 & 0 & 0 \\ 0 & 0 & 0 & 0 \end{pmatrix}_{\mu\nu}. \tag{S40}$$

Let $c^2 := \langle \cos^2 a \rangle$, $s^2 := \langle \sin^2 a \rangle$, and $t := \langle \cos a \sin a \rangle$. We have the inequality $t^2 \leq c^2 s^2$. So we just have to diagonalise the $2 \times 2$ matrix

$$m = \begin{pmatrix} c^2 & -it \\ +it & s^2 \end{pmatrix}. \tag{S41}$$

Define $\Delta = c^2 - s^2$. Then the eigenvalues are given by

$$\lambda_\pm = \frac{1 \pm \sqrt{\Delta^2 + 4t^2}}{2}, \tag{S42}$$

and the unnormalised eigenvectors are

$$v^\pm = \begin{pmatrix} \frac{\Delta \pm \sqrt{\Delta^2 + 4t^2}}{2t} \\ i \end{pmatrix}. \tag{S43}$$

Let $w_0^\pm = v_0^\pm / \|v^\pm\|$ and $w_1 = 1/\|v^\pm\|$.

This results in two non-zero eigenvalues $\lambda_\pm$ and the corresponding eigenvectors $v^\pm$, so that

$$\mathcal{E}[\cdot] = \sum_\pm E_\pm \cdot E_\pm^\dagger, \tag{S44}$$

where $E_\pm = \sqrt{\lambda_\pm}(w_0^\pm I + iw_1^\pm X)$. Now we can construct a unitary $U$ such that $UE_\pm \propto P_\pm$, i.e. a Pauli. Writing $U = c_\phi + is_\phi Z$, we have

$$UE_\pm = \sqrt{\lambda_\pm}(c_\phi + is_\phi X)(w_0^\pm I + iw_1^\pm X) = \sqrt{\lambda_\pm}(c_\phi w_0^\pm - s_\phi w_1^\pm)I + i\sqrt{\lambda_\pm}(s_\phi w_0^\pm + c_\phi w_1^\pm)X = \sqrt{\lambda_\pm}\alpha_0^\pm + i\sqrt{\lambda_\pm}\alpha_1^\pm X. \quad (S45)$$

Now let us choose $\phi$ such that we have both $\alpha_1^- = 0$ and $\alpha_0^+ = 0$. Then we need

$$s_\phi v_0^+ + c_\phi = 0, \quad c_\phi v_0^- - s_\phi = 0 \quad (S46)$$
$$\implies t_\phi v_0^+ + 1 = 0, \quad v_0^- - t_\phi = 0. \quad (S47)$$

First, let's check that there can be a solution. Rearranging for $t_\phi$, we need the requirement $v_0^- + 1/v_0^+ = 0$, which can be checked to hold. We can then find

$$\phi = \arctan(v_0^-) = \arctan\left(\frac{\Delta + \sqrt{\Delta^2 + 4t^2}}{2t}\right). \quad (S48)$$

Which leads to

$$UE_- = \sqrt{\lambda_-} I, \quad UE_+ = i\sqrt{\lambda_+} X. \quad (S49)$$

We can then interpret $UE_- = (1 - \tilde{p})I, UE_+ = i\tilde{p}X$, where

$$\tilde{p} = \sqrt{\lambda_+} = \frac{1 - \sqrt{\Delta^2 + 4t^2}}{2}, \quad (S50)$$

i.e. a bit-flip error channel.

## VII. OPTIMAL RECOVERY SCHEMES FOR QUBIT CSS CODES

In the following subsections, we will study optimal recovery schemes for qubit CSS codes. Such codes are stabilised by $X$-type and $Z$-type stabilisers, and its code space projectors can be written in the form $\Pi = \mathbb{A}_V \mathbb{B}_P$. Here $V = \{v\}_v$ denote the set of all 'vertices' and $P = \{p\}_p$ denote the set of 'plaquettes', which are indeed the vertices and plaquettes for the case of the 2D $\mathbb{Z}_2$ toric or surface code. They are composed of vertex and plaquette projectors $\mathbb{A}_V = \prod_{v \in V} \mathbb{A}_v, \mathbb{B}_P = \prod_{p \in P} \mathbb{B}_p$, where $\mathbb{A}_v = (1 + X_v)/2, \mathbb{B}_p = (1 + Z_p)/2$, where $X_v = \prod_{i \in \text{nbh}(v)} X_i, Z_p = \prod_{i \in \text{nbh}(p)} Z_i$ where $\text{nbh}(\cdot)$ denotes the neighbourhood and $i$ is an 'edge'. The 'neighbourhood' can be non-local depending on the code, but we adopt the language of the toric code for simplicity [S7]. Let the code be comprised of $N$ physical qubits, storing $K$ logical qubits. $\ell = (\ell^{(1)}, \ell^{(2)}, \ldots, \ell^{(K)})$ denote the logical computational basis states. For example, in the 2D $\mathbb{Z}_2$ surface code with smooth and rough boundaries, $K = 1$ and $\ell$ is just one qubit. In the 2D $\mathbb{Z}_2$ toric code with doubly periodic boundary conditions, $K = 2$ and $\ell = (\ell^{(1)}, \ell^{(2)})$.

Such codes can be spanned by the 'anyon number basis' [S8]:

$$\left|\overline{\ell}, (\tilde{s}_x, \tilde{s}_z)\right\rangle := Z^{\gamma(\tilde{s}_x)} X^{\gamma(\tilde{s}_z)} \left|\overline{\ell}, (\vec{0}, \vec{0})\right\rangle. \quad (S51)$$

Here, $|\overline{\ell}, (\vec{0}, \vec{0})\rangle$ is a logical computation basis state (i.e. an eigenstate for $\bar{Z}^{(j)}$, the logical Pauli $Z$ operator on the $j$th logical sector) We put a tilde ($\sim$) to remind ourselves that $\tilde{s}_{z,(x)}$ is an equivalence class of $X(Z)$-type Pauli strings that give rise to the syndromes $s_{z,(x)}$ that are related by $X(Z)$-type stabilisers. $\gamma(\bar{s}_{z,(x)})$ is a representative of one such class, which we make a choice. For the case of the toric code (let us choose odd system size for simplicity), we can make the intuitive choice that they correspond to a string with the shortest path between the syndrome locations. If two Pauli strings give the same syndromes but are not related by stabilisers, then they differ by a logical operator. This is illustrated for the toric code below for an $X$-type error (and therefore $Z$-type stabiliser measurement).

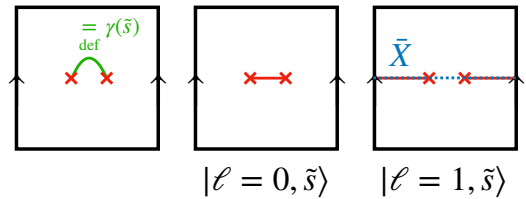

*General Pauli error.* We will consider two kinds of errors. First is a very general Pauli error,

$$\mathcal{E}_{\text{Pauli}}[\cdot] = \sum_P p(P) P(\cdot) P^\dagger. \tag{S52}$$

Here, no locality properties have been assumed. In the case where it is seperable into single-site channels, $\mathcal{E}_{\text{Pauli}} = \prod_i \mathcal{E}_i^{\text{Pauli}}$, each can be written as

$$\mathcal{E}_i^{\text{Pauli}}[\cdot] = p_I(\cdot) + p_X X(\cdot)X + p_Y Y(\cdot)Y + p_Z Z(\cdot)Z. \tag{S53}$$

In the case where each single-site error can be separated into bit-flip error followed by a phase-flip error, $\mathcal{E}_i^{\text{Pauli}} = \mathcal{E}_i^{\text{phase-flip}} \mathcal{E}_i^{\text{bit-flip}}$, where

$$\mathcal{E}_i^{\text{bit-flip}}[\cdot] = (1-q_X)(\cdot) + q_X X(\cdot)X, \tag{S54}$$

and similarly for phase-flip errors, we have $p_I = (1-q_X)(1-q_Z)$, $p_X = q_X(1-q_Z)$, $p_Y = q_X q_Z$, and $p_Z = (1-q_X)q_Z$.

We now state some useful observations for the Pauli error channel. After Pauli error is applied onto the Bell state and the error channel is dilated, the joint $QRE$ system can be written as

$$\left|\Psi'_{QRE}\right\rangle = \frac{1}{\sqrt{d}} \sum_k U_{QE} \left|\bar{k}\right\rangle_Q \left|k\right\rangle_R \left|\chi\right\rangle_E = \frac{1}{\sqrt{d}} \sum_k \sum_P \sqrt{p(P)} P_Q \left|\bar{k}\right\rangle_Q \left|k\right\rangle_R \left|P\right\rangle_E, \tag{S55}$$

where $\sigma = (\sigma_{\ell,i})_{\ell,i}$ are two boolean variables defined on every site, and we can decompose the Pauli string as $P = \prod_i Z_i^{\sigma_{i,z}} X_i^{\sigma_{i,x}}$, and $p(P)$ is the probability for a particular Pauli string error to occur.

Let us switch to the anyon number basis. $P_Q \left|\bar{k}\right\rangle_Q$ can be decomposed as

$$(-1)^{f(P)} Z^{\tilde{s}_x(P)} X^{\tilde{s}_z(P)} \bar{Z}^{\ell_x(P)} \bar{X}^{\ell_z(P)} \left|\bar{k}\right\rangle_Q = (-1)^{f(P)}(-1)^{k \oplus \ell_x(P) \oplus \ell_z(P)} \left|\overline{k \oplus \ell_z}, \tilde{s}\right\rangle. \tag{S56}$$

where $f(P) \in \{0,1\}$ is a phase term accumulated from commutation relations between $X$ and $Z$ Pauli operators. Therefore the state can be written as

$$\left|\Psi'_{RQE}\right\rangle = \frac{1}{\sqrt{d}} \sum_k \left|k\right\rangle_R \sum_{\tilde{s},\ell} (-1)^{\ell_x(k \oplus \ell_z)} \sum_{P \in \tilde{s},\ell} (-1)^{f(P)} \left|\overline{k \oplus \ell_z}, \tilde{s}\right\rangle_Q \sqrt{p(P)} \left|P\right\rangle_E. \tag{S57}$$

*Amplitude damping error.* The second is the amplitude damping channel. We will consider those that can be decomposed into single-site channels $\mathcal{E} = \prod_i \mathcal{E}_i$, where the single-site channel acts as

$$\mathcal{E}_i^{\text{ad}}[\cdot] = K_0(\cdot)K_0^\dagger + K_1(\cdot)K_1^\dagger, \tag{S58}$$

where

$$K_0 = \begin{pmatrix} 1 & 0 \\ 0 & \sqrt{1-p} \end{pmatrix}, \quad K_1 = \begin{pmatrix} 0 & \sqrt{p} \\ 0 & 0 \end{pmatrix}. \tag{S59}$$

When dilated, it acts as

$$V_{QE\leftarrow Q}^i \left|\psi\right\rangle_Q = \left(\frac{1+\sqrt{1-p}}{2} I + \frac{1-\sqrt{1-p}}{2} Z\right) \left|\psi\right\rangle_Q \otimes \left|0\right\rangle_E + \frac{\sqrt{p}}{2}(X + ZX) \left|\psi\right\rangle_Q \otimes \left|1\right\rangle_E$$
$$=: U_{QE} \left|\psi\right\rangle_Q \left|\chi\right\rangle_E. \tag{S60}$$

Now, consider Ising variables living on each site, $\sigma = (\sigma_i)_i$ and $\tau = (\tau_i)_i$. The action of the isometry on the whole system is then

$$V_{QE\leftarrow Q} \left|\psi\right\rangle_Q = \sum_\sigma \sum_\tau \sqrt{p(\sigma,\tau)} Z^\tau X^\sigma \left|\psi\right\rangle_Q \otimes \left|\sigma\right\rangle_E. \tag{S61}$$

$\sqrt{p(\sigma,\tau)}$ are non-negative real numbers that sum to $\sum_{\sigma,\tau} p(\sigma,\tau) = 1$. As in the Pauli case, $Z^\tau X^\sigma$ is some Pauli operator, and therefore its action on a logical computational basis state is

$$Z^\tau X^\sigma \left|\bar{k}\right\rangle = (-1)^{f(\sigma,\tau) + \ell_x(\tau)(k \oplus \ell_z(\sigma))} \left|\overline{k \oplus l_z(\sigma)}, \tilde{s}_x(\tau), \tilde{s}_z(\sigma)\right\rangle. \tag{S62}$$

The Kraus operator for the whole channel $\mathcal{E}$ can be written as

$$K_\sigma = \sum_\tau \sqrt{p_{\sigma\tau}} Z^\tau X^\sigma. \tag{S63}$$

## A. Pauli error and the Petz recovery channel

Let us first calculate

$$\mathcal{E}(\Pi) = \sum_{k,P} p(P) P \ket{\bar{k}}\bra{\bar{k}} P^\dagger = \sum_{k,P} p(P) \ket{\overline{k \oplus \ell_z}, \tilde{s}}\bra{\overline{k \oplus \ell_z}, \tilde{s}} \tag{S64}$$

$$= \sum_k \sum_{\ell, \tilde{s}} p(\ell, \tilde{s}) \ket{\overline{k \oplus \ell_z}, \tilde{s}}\bra{\overline{k \oplus \ell_z}, \tilde{s}} \tag{S65}$$

$$= \sum_{\ell_x} \sum_{k, \ell_z, \tilde{s}} p(\ell_z, \tilde{s}) \ket{\overline{k \oplus \ell_z}, \tilde{s}}\bra{\overline{k \oplus \ell_z}, \tilde{s}} \stackrel{k \to k \oplus \ell_z}{=} \sum_{\tilde{s}} p(\tilde{s}) \ket{\bar{k}, \tilde{s}}\bra{\bar{k}, \tilde{s}} \tag{S66}$$

$$\implies \mathcal{E}(\Pi)^{-1/2} = \sum_{\tilde{s}} \frac{1}{\sqrt{p(\tilde{s})}} \ket{\bar{k}, \tilde{s}}\bra{\bar{k}, \tilde{s}}. \tag{S67}$$

Now construct recovery channel Kraus operator $R_P$ for $\mathcal{R}[\cdot] = \sum_P R_P (\cdot) R_P^\dagger$:

$$R_P \equiv \Pi \sqrt{p(P)} P^\dagger \mathcal{E}(\Pi)^{-1/2} \tag{S68}$$

$$= \sqrt{p(P)} \sum_k \ket{\bar{k}}\bra{\overline{k \oplus \ell_z(P)}, \tilde{s}(P)} (-1)^{f(P)} (-1)^{k \oplus \ell_z(P) \oplus \ell_x(P)} \tag{S69}$$

$$\times \sum_{\tilde{s}} \frac{1}{\sqrt{p(\tilde{s})}} \sum_{\ell_z} \ket{\overline{\ell_z}, \tilde{s}}\bra{\overline{\ell_z}, \tilde{s}}. \tag{S70}$$

Inner products give $\tilde{s} = \tilde{s}(P)$, $k \oplus \ell_z(P) = \ell_z$. Then we have:

$$R_P = \sqrt{p(P)} \sum_k \ket{\bar{k}}\bra{\overline{k \oplus \ell_z(P)}, \tilde{s}(P)} \frac{1}{\sqrt{p(\tilde{s}(P))}} (-1)^{f(P)} (-1)^{k \oplus \ell_z(P) \oplus \ell_x(P)}. \tag{S71}$$

Then

$$\mathcal{R}[\cdot] = \sum_P \frac{p(P)}{p(\tilde{s}(P))} \sum_k \ket{\bar{k}}\bra{\overline{k \oplus \ell_z(P)}, \tilde{s}(P)} (-1)^{k \oplus \ell_z(P) \oplus \ell_x(P)} \tag{S72}$$

$$\times (\cdot) \sum_{k'} \ket{\overline{k' \oplus \ell_z(P)}, \tilde{s}(P)}\bra{\bar{k'}} (-1)^{k' \oplus \ell_z(P) \oplus \ell_x(P)} \tag{S73}$$

$$= \sum_{\ell^\circ, \tilde{s}^\circ} \frac{p(\ell^\circ, \tilde{s}^\circ)}{p(\tilde{s}^\circ)} \sum_k \ket{\bar{k}}\bra{\overline{k \oplus \ell_z^\circ}, \tilde{s}^\circ} (-1)^{k \oplus \ell_z^\circ \oplus \ell_x^\circ} (\cdot) \sum_{k'} \ket{\overline{k' \oplus \ell_z^\circ}, \tilde{s}^\circ}\bra{\bar{k'}} (-1)^{k' \oplus \ell_z^\circ \oplus \ell_x^\circ}. \tag{S74}$$

Now we can write

$$\sum_k \ket{\bar{k}}\bra{\overline{k \oplus \ell_z^\circ}, \tilde{s}^\circ} (-1)^{k \oplus \ell_z^\circ \oplus \ell_x^\circ} \stackrel{k \to k \oplus \ell_z^\circ}{=} \sum_k \ket{\overline{k \oplus \ell_z^\circ}}\bra{\bar{k}, \tilde{s}^\circ} (-1)^{k \oplus \ell_x^\circ}$$

$$= \bar{X}^{\ell_z^\circ} \left( \sum_k \ket{\bar{k}}\bra{\bar{k}} (-1)^{k \oplus \ell_x^\circ} \right) (XZ)^{\tilde{s}^\circ}$$

$$= \bar{X}^{\ell_z^\circ} \left( \sum_k \ket{\bar{k}}\bra{\bar{k}} (-1)^{k \oplus \ell_x^\circ} \right) \left( \sum_n \ket{\bar{n}}\bra{\bar{n}} \right) (XZ)^{\tilde{s}^\circ} \tag{S75}$$

$$= \bar{X}^{\ell_z^\circ} \bar{Z}^{\ell_x^\circ} \Pi (XZ)^{\tilde{s}^\circ}$$

$$= \bar{X}^{\ell_z^\circ} \bar{Z}^{\ell_x^\circ} (XZ)^{\tilde{s}^\circ} \underbrace{(ZX)^{\tilde{s}^\circ} \Pi (XZ)^{\tilde{s}^\circ}}_{=\Pi_{s^\circ}}$$

$$= \bar{X}^{\ell_z^\circ} \bar{Z}^{\ell_x^\circ} (XZ)^{\tilde{s}^\circ} \Pi_{s^\circ}.$$

An important fact to note here is that although we have chosen a specific representative $\gamma(\tilde{s})$ to define $\ket{\bar{\ell}, \tilde{s}}$, the projector $\Pi_{s^\circ} = (ZX)^{\tilde{s}^\circ} \Pi (XZ)^{\tilde{s}^\circ}$ is agnostic of this choice. This is because all the stabilisers in $\Pi$ commute apart from the endpoints of the string, which simply conjugates $Z_i \frac{1}{2}(1 + X_v) Z_i = \frac{1}{2}(1 - X_v)$ for $i \in \text{nbh}(v)$.

So we have:
$$\mathcal{R}[\cdot] = \sum_{\ell,\tilde{s}} p(\ell|\tilde{s}) \bar{X}^{\ell_z} \bar{Z}^{\ell_x} (XZ)^{\tilde{s}} \Pi_s(\cdot) \Pi_s (ZX)^{\tilde{s}} \bar{Z}^{\ell_x} \bar{X}^{\ell_z}. \tag{S76}$$

This can be interpreted in the following way. First the syndromes are measured ($\Pi_s(\cdot)\Pi_s$). The canonical $(XZ)^{\tilde{s}}$ is used to remove syndromes. Then, a probabilistic decoder is used to infer whether a logical operator should be used to correct a logical error or not.

### B. Bit-flip error and the SW recovery channel

As a warm-up, we first consider bit-flip errors. The state after error is
$$|\Psi'_{QRE}\rangle = \sum_k \sqrt{\lambda_k} U_{QE} |\bar{k}\rangle_Q |k\rangle_R |\chi\rangle_E = \sum_k \sqrt{\lambda_k} \sum_\sigma \sqrt{p(\sigma)} X_Q^\sigma |\bar{k}\rangle_Q |k\rangle_R |\sigma\rangle_E, \tag{S77}$$

where $\sigma = (\sigma_i)_i$, where the iteration is over all edges, and $X^\sigma = \prod_i X_i^{\sigma_i}$. Each $\sigma \in \{0,1\}$. Separate $X^\sigma$ into a syndrome and logical component, to get
$$|\Psi'_{QRE}\rangle = \sum_k \sqrt{\lambda_k} |k\rangle_R \sum_{s,\ell} |\overline{k \oplus \ell}, s\rangle_Q \sum_{\sigma \in s,\ell} \sqrt{p(\sigma)} |\sigma\rangle_E. \tag{S78}$$

Each of the logical-syndrome pair is orthogonal to each other.

Then the density matrix on subsystem $E$ is
$$\rho'_E = \sum_{\ell,\tilde{s}} \left( \sum_{\sigma \in \ell,\tilde{s}} \sqrt{p(\sigma)} |\sigma\rangle \right) \left( \sum_{\sigma' \in \ell,\tilde{s}} \sqrt{p(\sigma')} \langle\sigma| \right), \tag{S79}$$

which can be purified as
$$|\tilde{\psi}'_{QE}\rangle = \sum_{\ell,\tilde{s}} |\bar{\ell},\tilde{s}\rangle_Q \sum_{\sigma \in \ell,\tilde{s}} \sqrt{p(\sigma)} |\sigma\rangle. \tag{S80}$$

The density matrix for $\rho'_R = \sum_k \lambda_k |k\rangle\langle k|_R$ can be purified as
$$|\tilde{\psi}'_{R'}\rangle = \sum_k \sqrt{\lambda_k} |k\rangle_R |k\rangle_{R'}. \tag{S81}$$

Therefore the purification of $\rho'_R \otimes \rho'_E$ can be contained in $Q'RE$, $|\tilde{\Psi}_{Q'RE}\rangle = |\tilde{\psi}'_{R'}\rangle \otimes |\tilde{\psi}'_{QE}\rangle$, where $Q' = R'Q$. To match the Hilbert space, we add a dimension to $|\Psi_{Q'RE}\rangle := |\eta_{R'}\rangle \otimes |\Psi_{QRE}\rangle$, where $|\eta_{R'}\rangle$ is an arbitrary ancilla on $R'$.

Now the optimal purification is then determined by
$$F[\rho'_{RE}, \rho'_R \otimes \rho'_E] = \left| \langle\eta|_{R'} \langle\Psi'_{QRE}| U_{Q'}^{\text{argmax}} |\tilde{\Psi}\rangle_{Q'RE} \right|, \tag{S82}$$

where $U_{Q'}^{\text{argmax}} := \text{argmax}_{U_{Q'}} \left| \langle\eta|_{R'R''} \langle\Psi'_{QRE}| U_{Q'} |\tilde{\Psi}\rangle_{Q'RE} \right|$.

To determine $U_{Q'}^{\text{argmax}}$, notice that
$$F[\rho'_{RE}, \rho'_R \otimes \rho'_E] = \max_{U_{Q'}} \left| \text{tr}_{Q'RE} \left[ U_{Q'} |\tilde{\Psi}\rangle_{Q'RE} \langle\Psi'_{QRE}| \langle\eta|_{R'} \right] \right| \tag{S83}$$
$$= \max_{U_{Q'}} \left| \text{tr}_{Q'} \left[ U_{Q'} \, \text{tr}_{RE} \left\{ |\tilde{\Psi}\rangle_{Q'RE} \langle\Psi'_{QRE}| \langle\eta|_{R'} \right\} \right] \right|. \tag{S84}$$

Define the matrix $M_{Q'} = \text{tr}_{RE} \left\{ |\tilde{\Psi}\rangle_{Q'RE} \langle\Psi'_{QRE}| \langle\eta|_{R'} \right\}$. If the SVD of this matrix can be found, $M_{Q'} = W_Q \Sigma_Q V_{Q'}^\dagger$, then the choice $U_{Q'}^{\text{argmax}} = V_{Q'} W_{Q'}^\dagger$ will give the optimal fidelity, since $\Sigma_{Q'}$ is a positive-definite matrix. Now compute
$$M_{Q'} = \text{tr}_{RE} \left\{ |\tilde{\Psi}\rangle_{Q'RE} \langle\Psi'_{QRE}| \langle\eta|_{R'} \right\} \tag{S85}$$
$$= \text{tr}_{RE} \left\{ \sum_r \sqrt{\lambda_r} |r\rangle_R |r\rangle_{R'} \sum_{\ell,\tilde{s}} |\bar{\ell},\tilde{s}\rangle_Q \sum_{\sigma \in \ell,\tilde{s}} \sqrt{p(\sigma)} |\sigma\rangle_E \sum_k \sqrt{\lambda_k} \sum_{\ell',\tilde{s}'} \langle \overline{k \oplus \ell'}, \tilde{s}'|_Q \langle k|_R \sum_{\sigma' \in \ell',\tilde{s}'} \sqrt{p(\sigma')} \langle\sigma'|_E \langle\eta|_{R'} \right\}. \tag{S86}$$

Now the trace over $R$ sets $k = r$. The trace over $E$ gives $\sigma = \sigma'$, which is only possible if $s = s'$ and $\ell = \ell'$. Now $\sum_{\sigma \in \ell, \tilde{s}} p(\sigma) = p(\ell, \tilde{s})$. Therefore we would have

$$M_{Q'} = \sum_k \lambda_k \ket{k}_{R'} \sum_{\ell, \tilde{s}} p(\ell, \tilde{s}) \ket{\bar{\ell}, \tilde{s}}_Q \bra{k \oplus \ell, \tilde{s}}_Q \bra{\eta}_{R'}. \tag{S87}$$

At this point, how do we do the SVD of this object? First thing we can do is note that it is block diagonal in $s$. Separate $Q = Q_S Q_L$ (syndrome and logical space) to get

$$M_{Q'} = \sum_{\tilde{s}} \ket{\tilde{s}}\bra{\tilde{s}}_{Q_S} \sum_{k,\ell} \lambda_k \ket{k}_{R'} p(\ell, \tilde{s}) \ket{\bar{\ell}}_{Q_L} \bra{k \oplus \ell}_{Q_L} \bra{\eta}_{R'}. \tag{S88}$$

Now shift $\ell \to \ell \oplus k$ to find

$$M_{Q'} = \sum_{\tilde{s}} \ket{\tilde{s}}\bra{\tilde{s}}_{Q_S} \sum_{\ell} \left( \sum_k \lambda_k \ket{k}_{R'} p(\ell \oplus k, s) \ket{\overline{\ell \oplus k}}_{Q_L} \right) \bra{\bar{\ell}}_{Q_L} \bra{\eta}_{R'}. \tag{S89}$$

Now note that

$$\sum_k \lambda_k \ket{k}_{R'} p(\ell \oplus k, \tilde{s}) \ket{\overline{\ell \oplus k}}_{Q_L} = \text{CNOT}_{R' \to Q_L} \underbrace{\left( \sum_k \lambda_k p(\ell \oplus k, \tilde{s}) \ket{k}_{R'} \right)}_{:= \ket{\tilde{w}_{\ell,s}}_{R'}} \ket{\bar{\ell}}_{Q_L}. \tag{S90}$$

We have

$$\braket{\tilde{w}_{\ell,s} | \tilde{w}_{\ell',s}} = \sum_k \lambda_k^2 p(\ell \oplus k, s)^2. \tag{S91}$$

So, define $\mu_{\ell,s} = \sqrt{\sum_k \lambda_k^2 p(\ell \oplus k, s)^2}$ and $\ket{w_{\ell,s}}_{R'} = \mu_{\ell,s}^{-1} \ket{\tilde{w}_{\ell,s}}_{R'}$. So we have

$$M_{Q'} = \text{CNOT}_{R' \to Q_L} \sum_{\ell, \tilde{s}} \beta_{\ell, \tilde{s}} \ket{\tilde{s}}\bra{\tilde{s}}_{Q_S} \ket{w_{\ell,\tilde{s}}}\bra{\eta}_{R'} \ket{\bar{\ell}}\bra{\bar{\ell}}_{Q_L}, \tag{S92}$$

i.e. disregarding the CNOT, the matrix is entirely diagonal on $Q$. Now this is not a full rank matrix, since the rank on $R'$ is only 1. We have absorbed normalisation into the singular value $\beta_{l,\tilde{s}}$. Therefore we must extend the matrix to the space orthogonal to $\ket{w_{\ell,s}}$ and $\ket{\eta}$. Then we need the additional term

$$M_{Q'}^{\text{full rank}} = \text{CNOT}_{R' \to Q_L} \sum_{\ell, \tilde{s}} \sigma_{\ell, \tilde{s}} \ket{\tilde{s}}\bra{\tilde{s}}_{Q_S} \ket{w_{\ell,\tilde{s}}}\bra{\eta}_{R'} \ket{\bar{\ell}}\bra{\bar{\ell}}_{Q_L} + \text{CNOT}_{R' \to Q_L} \sum_{\ell, \tilde{s}} \ket{\tilde{s}}\bra{\tilde{s}}_{Q_S} \ket{w_{\ell,\tilde{s}}^\perp}\bra{\eta^\perp}_{R'} \ket{\bar{\ell}}\bra{\bar{\ell}}_{Q_L}. \tag{S93}$$

Then define the local unitary $V_{R'}^{(\ell, \tilde{s})} = \ket{\eta}\bra{w_{\ell,\tilde{s}}}_{R'} + \ket{\eta^\perp}\bra{w_{\ell,\tilde{s}}^\perp}_{R'}$. Finally, the optimal unitary is given by

$$U_{Q'}^{\text{argmax}} = \sum_{\ell, \tilde{s}} \ket{\bar{\ell}, \tilde{s}}\bra{\bar{\ell}, \tilde{s}}_Q \otimes V_{R'}^{(\ell, s)} \text{CNOT}_{R' \to Q_L}. \tag{S94}$$

Now let us apply this operator onto $\ket{\tilde{\Psi}_{Q'RE}}$:

$$U_{Q'}^{\text{argmax}} \ket{\tilde{\Psi}_{Q'RE}} = \sum_{\ell', \tilde{s}'} \ket{\bar{\ell}', \tilde{s}'}\bra{\bar{\ell}', \tilde{s}'}_Q \otimes V_{R'}^{(\ell', \tilde{s}')} \text{CNOT}_{R' \to Q_L} \sum_r \sqrt{\lambda_r} \ket{r}_R \ket{r}_{R'} \sum_{\ell, \tilde{s}} \ket{\bar{\ell}, \tilde{s}}_Q \sum_{\sigma \in \ell, \tilde{s}} \sqrt{p(\sigma)} \ket{\sigma}_E \tag{S95}$$

$$= \sum_r V_{R'}^{(\ell \oplus r, \tilde{s})} \sqrt{\lambda_r} \ket{r}_R \ket{r}_{R'} \sum_{\ell, \tilde{s}} \ket{\overline{\ell \oplus r}, \tilde{s}}_Q \sum_{\sigma \in \tilde{s}, \ell} \sqrt{p(\sigma)} \ket{\sigma}_E \tag{S96}$$

$$= \sum_{r, \tilde{s}, \ell} \sqrt{\lambda_r p(\tilde{s}, \ell)} \ket{r}_R \left( V_{R'}^{(\ell \oplus r, \tilde{s})} \ket{r}_{R'} \ket{\overline{\ell \oplus r}, \tilde{s}}_Q \right) \sum_{\sigma \in \tilde{s}, \ell} \sqrt{\frac{p(\sigma)}{p(\tilde{s}, \ell)}} \ket{\sigma}_E. \tag{S97}$$

Now let's check that $|\phi_{Q'}^{r,(\tilde{s},\ell)}\rangle := V_{R'}^{(\ell\oplus r,\tilde{s})}|r\rangle_{R'}|\overline{\ell\oplus r},\tilde{s}\rangle_Q$ is orthogonal:

$$\langle\phi_{Q'}^{r,(\tilde{s},\ell)}|\phi_{Q'}^{r',(\tilde{s}',\ell')}\rangle = \langle r|_{R'}V_{R'}^{(\ell\oplus r,\tilde{s})\dagger}\langle\overline{\ell\oplus r},\tilde{s}|_Q V_{R'}^{(\ell'\oplus r',\tilde{s}')}|r'\rangle_{R'}|\overline{\ell'\oplus r'},\tilde{s}'\rangle_Q \tag{S98}$$

$$= \delta_{\tilde{s},\tilde{s}'}\delta_{\ell\oplus r,\ell'\oplus r'}\langle r|_{R'}V_{R'}^{(\ell\oplus r,\tilde{s})\dagger}V_{R'}^{(\ell'\oplus r',\tilde{s}')}|r'\rangle_{R'} \tag{S99}$$

$$= \delta_{\tilde{s},\tilde{s}'}\delta_{\ell\oplus r,\ell'\oplus r'}\langle r|_{R'}V_{R'}^{(\ell\oplus r,\tilde{s})\dagger}V_{R'}^{(\ell\oplus r,\tilde{s})}|r'\rangle_{R'} \tag{S100}$$

$$= \delta_{\tilde{s},\tilde{s}'}\delta_{\ell\oplus r,\ell'\oplus r'}\delta_{rr'} = \delta_{\tilde{s},\tilde{s}'}\delta_{\ell,\ell'}\delta_{r,r'}. \tag{S101}$$

Importantly, $\{|k_R\rangle\}_k$, $\{|\phi_{Q'}^{k\ell}\rangle\}$ and $\{|l_E\rangle\}$ now independently form an orthonormal set. So the state is now written in the form

$$|\tilde{\Psi}_{Q'RE}\rangle = \sum_{k\ell}\sqrt{\mu_k\nu_l}|k_R\rangle\otimes|\phi_{Q'}^{k\ell}\rangle\otimes|l_E\rangle. \tag{S102}$$

For us, $l\leftrightarrow(\ell,\tilde{s})$ and $k\leftrightarrow r$. Then, the measurement basis is given by $\mathbb{\Pi}_{kl} = |\phi_{Q'}^{kl}\rangle\langle\phi_{Q'}^{kl}|$, which for us is

$$\mathbb{\Pi}_{\tilde{s},\ell} = \sum_r V_{R'}^{(\ell\oplus r,\tilde{s})}|r\rangle_{R'}|\overline{\ell\oplus r},\tilde{s}\rangle_Q\langle r|_{R'}V_{R'}^{(\ell\oplus r,\tilde{s})\dagger}\langle\overline{\ell\oplus r},\tilde{s}|_Q = \mathbb{\Pi}_{\ell|\tilde{s}}\mathbb{\Pi}_s. \tag{S103}$$

Note how, if we trace over $R'$ in the projector above, Now we must come up with a conditional unitary recovery channel that restores $Q$:

$$U^{(s,\ell)}|\phi_{Q'}^{r,(s,\ell)}\rangle = U^{(s,\ell)}V_{R'}^{(\ell\oplus r,s)}|r\rangle_{R'}|\overline{\ell\oplus r},s\rangle_Q\langle r|_{R'} = |\bar{r}\rangle_Q|\eta\rangle_{R'}. \tag{S104}$$

The clear answer is $U^{(s,\ell)} = \bar{X}^{\tilde{s}}\bar{X}^\ell$. But we also need to erase the entanglement with $R'$, and therefore apply the following controlled unitary channel following $U^{(s,\ell)}$:

$$\sum_r|\bar{r}\rangle\langle\bar{r}|_Q\otimes\left(U_{\eta\leftarrow r}V_{R'}^{(\ell\oplus r,s)\dagger}\right), \tag{S105}$$

where $U_{\eta\leftarrow r}$ is any unitary that rotates $r$ to $\eta$ (the initially prepared state on $R'$). Therefore the recovery unitary is

$$U_{Q'}^{(\ell,s)} = \sum_r|\bar{r}\rangle\langle\bar{r}|_Q\otimes\left(U_{\eta\leftarrow r}V_{R'}^{(\ell\oplus r,s)\dagger}\right)X_Q^{\tilde{s}}\bar{X}_Q^\ell. \tag{S106}$$

So, the Kraus operator for the recovery channel can be written as

$$U_{Q'}^{(\ell,s)}\mathbb{\Pi}_{\ell,\tilde{s}} = \sum_r|\eta\rangle_{R'}|\bar{r},0\rangle_Q\langle r|_{R'}V_{R'}^{(\ell\oplus r,s)\dagger}\langle\overline{\ell\oplus r},s|_Q \tag{S107}$$

$$\tag{S108}$$

Therefore the channel on $R'Q$ is

$$\mathcal{R}_{R'Q}[\cdot] = \sum_{\ell,s}U_{Q'}^{(\ell,s)}\mathbb{\Pi}_{\ell,\tilde{s}}(\cdot)\mathbb{\Pi}_{\ell,\tilde{s}}U_{Q'}^{(\ell,s)\dagger} \tag{S109}$$

To make this channel be exclusively on $Q$, we should prepare an initial state $|\eta\rangle\langle\eta|_{R'}$ on $R'$ (as that is what we did for $|\Psi'_{RQE}\rangle$) and trace it out:

$$\mathcal{R}_Q[\cdot] = \text{tr}_{R'}\sum_{\ell,s}U_{Q'}^{(\ell,s)}\mathbb{\Pi}_{\ell,\tilde{s}}(\cdot\otimes|\eta\rangle\langle\eta|_{R'})\mathbb{\Pi}_{\ell,\tilde{s}}U_{Q'}^{(\ell,s)\dagger}. \tag{S110}$$

Now consider

$$U_{Q'}^{(\ell,s)}\mathbb{\Pi}_{\ell,\tilde{s}}|\eta\rangle_{R'} = \sum_r|\eta\rangle_{R'}|\bar{r},0\rangle_Q\langle r|_{R'}V_{R'}^{(\ell\oplus r,s)\dagger}|\eta\rangle_{R'}\langle\overline{\ell\oplus r},\tilde{s}|_Q \tag{S111}$$

$$= \sum_r|\eta\rangle_{R'}|\bar{r},0\rangle_Q\langle r|w_{\ell\oplus r,\tilde{s}}\rangle_{R'}\langle\overline{\ell\oplus r},\tilde{s}|_Q \tag{S112}$$

$$= \sum_r\frac{\lambda_{\ell\oplus r}p(\ell,\tilde{s})}{\mu_{\ell\oplus r,\tilde{s}}}|\eta\rangle_{R'}|\bar{r},0\rangle_Q\langle\overline{\ell\oplus r},\tilde{s}|_Q \tag{S113}$$

$$\underset{r\to r\oplus\ell}{=}\sum_r\frac{\lambda_r p(\ell,\tilde{s})}{\mu_{r,\tilde{s}}}|\eta\rangle_{R'}|\overline{r\oplus\ell},0\rangle_Q\langle\bar{r},\tilde{s}|_Q. \tag{S114}$$

Then $\mathcal{R}_Q$ is given by

$$\mathcal{R}_Q[\cdot] = \sum_{\ell,s} \left( \sum_r \frac{\lambda_r p(\ell,\tilde{s})}{\mu_{r,\tilde{s}}} \left|\overline{r \oplus \ell}, 0\right\rangle\!\left\langle\overline{r}, \tilde{s}\right|_Q \right) \cdot \left( \sum_{r'} \frac{\lambda_{r'} p(\ell,\tilde{s})}{\mu_{r',\tilde{s}}} \left|\overline{r'}, \tilde{s}\right\rangle\!\left\langle\overline{r' \oplus \ell}, 0\right|_Q \right) \tag{S115}$$

$$= \sum_{\ell,s} \bar{X}^\ell \left( \sum_r \frac{\lambda_r p(\ell,\tilde{s})}{\mu_{r,\tilde{s}}} \left|\overline{r}, 0\right\rangle\!\left\langle\overline{r}, \tilde{s}\right|_Q \right) \cdot \left( \sum_{r'} \frac{\lambda_{r'} p(\ell,\tilde{s})}{\mu_{r',\tilde{s}}} \left|\overline{r'}, \tilde{s}\right\rangle\!\left\langle\overline{r'}, 0\right|_Q \right) \bar{X}^\ell. \tag{S116}$$

At this point, we specialise to $\lambda_r = 1/d$. Then $\mu_{\ell,\tilde{s}} = \sqrt{\sum_k p(\ell \oplus k, \tilde{s})^2}/d = \sqrt{\sum_k p(k, \tilde{s})^2}/d$, which is independent of $\ell$. Therefore we get

$$\mathcal{R}_Q[\cdot] = \sum_{\ell,s} \frac{p(\ell,\tilde{s})^2}{\sum_{\ell'} p(\ell',\tilde{s})^2} \bar{X}^\ell \left( \sum_r \left|\overline{r}, 0\right\rangle\!\left\langle\overline{r}, \tilde{s}\right|_Q \right) \cdot \left( \sum_{r'} \left|\overline{r'}, \tilde{s}\right\rangle\!\left\langle\overline{r'}, 0\right|_Q \right) \bar{X}^\ell. \tag{S117}$$

Using the same trick as Eq. (S75), we find:

$$\mathcal{R}_Q[\cdot] = \sum_{\ell,\tilde{s}} \frac{p(\ell,\tilde{s})^2}{\sum_{\ell'} p(\ell',\tilde{s})^2} \bar{X}^\ell X^{\tilde{s}} \Pi_s(\cdot) \Pi_s X^{\tilde{s}} \bar{X}^\ell, \tag{S118}$$

which can be written as

$$\mathcal{R}_Q[\cdot] = \sum_{\ell,\tilde{s}} q(\ell|\tilde{s}) \bar{X}^\ell X^{\tilde{s}} \Pi_s(\cdot) \Pi_s X^{\tilde{s}} \bar{X}^\ell, \tag{S119}$$

where

$$q(\ell|\tilde{s}) = \frac{p(\ell,\tilde{s})^2}{\sum_{\ell'} p(\ell',\tilde{s})^2}. \tag{S120}$$

### C. Pauli error and the SW recovery channel

We now consider general Pauli errors. We assume a uniform prior logical state $\lambda_k = 1/d$. The state after error is

$$\left|\Psi'_{QRE}\right\rangle = \frac{1}{\sqrt{d}} \sum_k U_{QE} \left|\bar{k}\right\rangle_Q \left|k\right\rangle_R \left|0\right\rangle_E = \frac{1}{\sqrt{d}} \sum_k \sum_P \sqrt{p(P)} P_Q \left|\bar{k}\right\rangle_Q \left|k\right\rangle_R \left|P\right\rangle_E. \tag{S121}$$

Here $d = 2^K$, where $K$ is the number of logical qubits. The action of the Pauli string $P$ on the code space is $P\left|\bar{k}, 0\right\rangle_Q = (-1)^{f(P)}(-1)^{(k\oplus \ell_z)\cdot \ell_x} \left|\overline{k \oplus \ell_z}, \tilde{s}\right\rangle_Q$.

Separating $P$ into these components, we get

$$\left|\Psi'_{QRE}\right\rangle = \frac{1}{\sqrt{d}} \sum_k \left|k\right\rangle_R \sum_{\ell,\tilde{s}} (-1)^{(k\oplus \ell_z)\cdot \ell_x} \left|\overline{k \oplus \ell_z}, \tilde{s}\right\rangle_Q \sum_{P \in \ell,\tilde{s}} (-1)^{f(P)} \sqrt{p(P)} \left|P\right\rangle_E. \tag{S122}$$

Each of the logical-syndrome pairs is orthogonal to each other.

Then the density matrix on subsystem $E$ is found by tracing out $R$ and $Q$. Let the indices on the ket be primed. Then the trace on $R$ sets $k = k'$, and the one on $Q$ sets $\ell_z = \ell'_z$ and $\tilde{s} = \tilde{s}'$. The remaining sum on $k$ gives

$$\frac{1}{d} \sum_k (-1)^{(\ell_x + \ell'_x)(k \oplus \ell_z)} = \delta_{\ell_x, \ell'_x}. \tag{S123}$$

Therefore we get

$$\rho'_E = \sum_{\ell,\tilde{s}} \underbrace{\left( \sum_{P \in \ell,\tilde{s}} (-1)^{f(P)} \sqrt{p(P)} \left|P\right\rangle \right)}_{:=\left|e^{\ell s}\right\rangle_E} \left(\text{h.c.}\right). \tag{S124}$$

Then the purification of $\rho'_R \otimes \rho'_E$ is

$$\left|\tilde{\Psi}'_{Q'RE}\right\rangle = \frac{1}{\sqrt{d}} \sum_r |r\rangle_R |r\rangle_{R'} \sum_{\ell s} |\overline{\ell_z}, s\rangle_Q |\ell_x\rangle_{R'} |e^{\ell s}\rangle_E. \tag{S125}$$

Let $|\Psi'_{Q'RE}\rangle := |\Psi'_{QRE}\rangle |\eta\rangle_{R'R''}$. Then as before, consider the matrix

$$M_{Q'} = \text{tr}_{RE}\left\{\left|\tilde{\Psi}'_{Q'RE}\right\rangle\!\!\left\langle\Psi'_{Q'RE}\right|\right\} \tag{S126}$$

$$= \frac{1}{d}\text{tr}_{RE}\left\{\sum_r |r\rangle_R |r\rangle_{R'} \sum_{\tilde{s},\ell_x,\ell_z} |\ell_x\rangle_{Q_X} |\overline{\ell_z},\tilde{s}\rangle_Q \sum_{P \in \tilde{s},\ell_x,\ell_z} \sqrt{p(P)} |P\rangle_E \right.$$

$$\left.\times \sum_k \sum_{\tilde{s}',\ell'_x,\ell'_z} \langle\overline{k \oplus \ell'_z},\tilde{s}'|_Q \langle k|_R \sum_{P' \in \tilde{s}',\ell'_x,\ell'_z} \sqrt{p(P')} (-1)^{(k \oplus \ell'_z) \cdot \ell'_x} \langle P'|_E \langle\eta|_{R'R''}\right\}. \tag{S127}$$

The trace over $R$ sets $k = r$. The trace over $E$ gives $P = P'$, which is only possible if $\tilde{s} = \tilde{s}'$, $\ell_x = \ell'_x$, and $\ell_z = \ell'_z$. Now $\sum_{P \in \tilde{s},\ell_x,\ell_z} p(P) = p(\ell_x,\ell_z,\tilde{s})$. Therefore we have

$$M_{Q'} = \frac{1}{d}\sum_k |k\rangle_{R'} \sum_{\tilde{s},\ell_x,\ell_z} p(\ell_x,\ell_z,\tilde{s})(-1)^{(k \oplus \ell_z) \cdot \ell_x} |\ell_x\rangle_{Q_X} |\overline{\ell_z},\tilde{s}\rangle_Q \langle\overline{k \oplus \ell_z},\tilde{s}|_Q \langle\eta,\eta_x|_{R'R''}. \tag{S128}$$

We note that it is block diagonal in $\tilde{s}$. Separate $Q = Q_S Q_L$ to get

$$M_{Q'} = \frac{1}{d}\sum_{\tilde{s}} |\tilde{s}\rangle\langle\tilde{s}|_{Q_S} \sum_{\ell_z}\left(\sum_{k,\ell_x} |k\rangle_{R'} |\ell_x\rangle_{Q_X} p(\ell_x,\ell_z,\tilde{s})(-1)^{(k \oplus \ell_z) \cdot \ell_x} |\overline{\ell_z}\rangle_{Q_L}\right) \langle\overline{k \oplus \ell_z}|_{Q_L} \langle\eta,\eta_x|_{R'R''}. \tag{S129}$$

Now shift $\ell_z \to \ell_z \oplus k$ to find

$$M_{Q'} = \frac{1}{d}\sum_{\tilde{s}} |\tilde{s}\rangle\langle\tilde{s}|_{Q_S} \sum_{\ell_z}\left(\sum_{k,\ell_x} |k\rangle_{R'} |\ell_x\rangle_{Q_X} p(\ell_x,\ell_z \oplus k,\tilde{s})(-1)^{\ell_z \cdot \ell_x} |\overline{\ell_z \oplus k}\rangle_{Q_L}\right) \langle\overline{\ell_z}|_{Q_L} \langle\eta,\eta_x|_{R'R''}. \tag{S130}$$

Now note that we can factor out the bit-flip with a CNOT, and the phase-flip with a CZ:

$$\sum_{k,\ell_x} p(\ell_x,\ell_z \oplus k,\tilde{s})(-1)^{\ell_z \cdot \ell_x} |k\rangle_{R'} |\ell_x\rangle_{Q_X} |\overline{\ell_z \oplus k}\rangle_{Q_L}$$

$$= \text{CNOT}_{R' \to Q_L} \text{CZ}_{Q_L,Q_X} \underbrace{\left(\sum_{k,\ell_x} p(\ell_x,\ell_z \oplus k,\tilde{s}) |k\rangle_{R'} |\ell_x\rangle_{Q_X}\right)}_{:=|\tilde{w}_{\ell_z,\tilde{s}}\rangle_{R'R''}} |\overline{\ell_z}\rangle_{Q_L}. \tag{S131}$$

We evaluate the norm:

$$\langle\tilde{w}_{\ell_z,\tilde{s}}|\tilde{w}_{\ell_z,\tilde{s}}\rangle = \sum_{k,\ell_x} p(\ell_x,\ell_z \oplus k,\tilde{s})^2 = \sum_{k',\ell_x} p(\ell_x,k',\tilde{s})^2 =: \mu_{\tilde{s}}^2. \tag{S132}$$

Define $|w_{\ell_z,\tilde{s}}\rangle_{R'Q_X} = \mu_{\tilde{s}}^{-1} |\tilde{w}_{\ell_z,\tilde{s}}\rangle_{R'Q_X}$. So we have

$$M_{Q'} = \text{CNOT}_{R' \to Q_L} CZ_{Q_L,R''} \sum_{\ell_z,\tilde{s}} \beta_{\tilde{s}} |\tilde{s}\rangle\langle\tilde{s}|_{Q_S} |w_{\ell_z,\tilde{s}}\rangle\langle\eta|_{R'R''} |\overline{\ell_z}\rangle\langle\overline{\ell_z}|_{Q_L}, \tag{S133}$$

where we have absorbed normalizations into $\beta_{\tilde{s}}$. Extending to a full rank unitary exactly as before via the orthogonal padding $V_{R'R''}^{(\ell_z,\tilde{s})} = |\eta,\eta_x\rangle\langle w_{\ell_z,\tilde{s}}|_{R'Q_X} + \ldots$, the optimal unitary is given by

$$U_{Q'}^{\text{argmax}} = \sum_{\ell_z,\tilde{s}} |\overline{\ell_z},\tilde{s}\rangle\langle\overline{\ell_z},\tilde{s}|_Q \otimes V_{R'R''}^{(\ell_z,\tilde{s})} CZ_{Q_L,R''} \text{CNOT}_{R' \to Q_L}. \tag{S134}$$

Now let us apply this operator onto $|\tilde{\Psi}_{Q'RE}\rangle$:

$$U_{Q'}^{\text{argmax}}|\tilde{\Psi}_{Q'RE}\rangle = \sum_{\ell'_z,\tilde{s}'}|\bar{\ell}'_z,\tilde{s}'\rangle\langle\bar{\ell}'_z,\tilde{s}'|_Q \otimes V_{R'Q_X}^{(\ell'_z,\tilde{s}')}\text{CNOT}_{R'\to Q_L}\text{CZ}_{Q_L,Q_X} \tag{S135}$$

$$\times \sum_r \frac{1}{\sqrt{d}}|r\rangle_R|r\rangle_{R'}\sum_{\ell_x,\tilde{s}}|\ell_x\rangle_{Q_X}\sum_{\ell_z}|\overline{\ell_z},\tilde{s}\rangle_Q \sum_{P\in\tilde{s},\ell_x,\ell_z}\sqrt{p(P)}|P\rangle_E$$

$$= \frac{1}{\sqrt{d}}\sum_{r\ell\tilde{s}}\underbrace{(-1)^{(\ell_z\oplus r)\ell_x}V_{R'R''}^{(\ell_z\oplus r,\tilde{s})}|r\rangle_{R'}|\ell_x\rangle_{R''}|\overline{\ell_z\oplus r},\tilde{s}\rangle_Q}_{=:|\phi^{r,(\ell,\tilde{s})}\rangle_{Q'}}|e^{\ell\tilde{s}}\rangle_E. \tag{S136}$$

Now let's check that $|\phi_{Q'}^{r,(\ell,\tilde{s})}\rangle$ is orthogonal:

$$\left\langle\phi_{Q'}^{r,(\ell,\tilde{s})}\Big|\phi_{Q'}^{r',(\ell',\tilde{s}')}\right\rangle. \tag{S137}$$

Let the indices on the ket be primed. The inner product on $Q$ sets $s = s'$ and $\ell_z \oplus r = \ell'_z \oplus r'$. Then unitaries cancel and gives $r = r'$ and $\ell_x = \ell'_x$ on $R'R''$. Therefore they are orthonormal. Therefore the measurement is with respect to these states:

$$\mathbb{\Pi}_{\ell s} = \sum_r \left|\phi^{r,(\ell,\tilde{s})}\right\rangle\!\!\left\langle\phi^{r,(\ell,\tilde{s})}\right|_{Q'}. \tag{S138}$$

Now construct a recovery unitary for $|\phi^{r,(\ell,\tilde{s})}\rangle = (-1)^{(\ell_z\oplus r)\ell_x}V_{R'R''}^{(\ell_z\oplus r,\tilde{s})}|r\rangle_{R'}|\ell_x\rangle_{R''}|\overline{\ell_z\oplus r},\tilde{s}\rangle_Q$. First, we should undo the syndromes with $(XZ)^{\tilde{s}}$. Then, we undo the phase with $\bar{Z}^{\ell_x}$ and the bit-flip $\bar{X}^{\ell_z}$. This results in

$$V_{R'R''}^{(\ell_z\oplus r,\tilde{s})}|r\rangle_{R'}|\ell_x\rangle_{R''}|\bar{r}\rangle_Q. \tag{S139}$$

Follow this up with a controlled unitary:

$$\sum_r |\bar{r}\rangle\!\langle\bar{r}|_Q \otimes U_{R'R''}^{\eta\leftarrow r\ell_x}V_{R'R''}^{(\ell_z\oplus r,\tilde{s})\dagger}. \tag{S140}$$

So the recovery unitary is

$$U^{\ell\tilde{s}} = \sum_r |\bar{r}\rangle\!\langle\bar{r}|_Q \otimes U_{R'R''}^{\eta\leftarrow r\ell_x}V_{R'R''}^{(\ell_z\oplus r,\tilde{s})\dagger}\bar{X}^{\ell_z}\bar{Z}^{\ell_x}(XZ)^{\tilde{s}}. \tag{S141}$$

Now we trace out $R'R''$ to define a channel only on $Q$. We have:

$$\mathcal{R}_{Q'}[\cdot] = \sum_{\ell\tilde{s}} U^{\ell\tilde{s}}\mathbb{\Pi}_{\ell\tilde{s}}(\cdot)\mathbb{\Pi}_{\ell\tilde{s}}U^{\ell\tilde{s}\dagger} \tag{S142}$$

Then the channel on $Q$ is

$$\mathcal{R}_Q[\cdot] = \sum_{\ell\tilde{s}} U^{\ell\tilde{s}}\mathbb{\Pi}_{\ell\tilde{s}}(\cdot \otimes |\eta\rangle\!\langle\eta|_{R'R''})\mathbb{\Pi}_{\ell\tilde{s}}U^{\ell\tilde{s}\dagger}. \tag{S143}$$

Let us look at $U^{\ell\tilde{s}}\mathbb{\Pi}_{\ell\tilde{s}}|\eta\rangle_{R'R''}$. We have

$$\Pi_{ls} = \sum_r |\overline{\ell_z\oplus r},\tilde{s}\rangle\langle\overline{\ell_z\oplus r},\tilde{s}|_Q V_{R'R''}^{(\ell_z\oplus r,\tilde{s})}|r,\ell_x\rangle\langle r,\ell_x|_{R'R''}V_{R'R''}^{(\ell_z\oplus r,\tilde{s})\dagger},$$

And

$$\langle r,\ell_x|V_{R'R''}^{(\ell_z\oplus r,\tilde{s})\dagger}|\eta\rangle_{R'R''} = \langle r,\ell_x|w_{\ell_z\oplus r,\tilde{s}}\rangle = \frac{p(\ell_x,\ell_z,\tilde{s})}{\mu_s}.$$

So

$$\mathbb{\Pi}_{\ell\tilde{s}}|\eta\rangle_{R'R''} = p(\ell,\tilde{s})\sum_r |\overline{\ell_z\oplus r},\tilde{s}\rangle\langle\overline{\ell_z\oplus r},\tilde{s}|_Q V_{R'R''}^{(\ell_z\oplus r,\tilde{s})}|r,\ell_x\rangle_{R'R''}.$$

Now consider

$$U_{\text{rec}}^{ls}\Pi_{ls}\ket{\eta}_{R'R''} = \sum_r \ket{\bar{r}}\bra{\bar{r}}_Q \otimes U_{R'R''}^{r \leftarrow r \oplus \ell_x} V_{R'R''}^{(\ell_z \oplus r, \tilde{s})\dagger} \bar{X}^{\ell_z} \bar{Z}^{\ell_x}(XZ)^{\tilde{s}}$$

$$\times p(\ell, \tilde{s}) \sum_{r'} \ket{\overline{\ell_z \oplus r'}, \tilde{s}} \bra{\overline{\ell_z \oplus r'}, \tilde{s}}_Q V_{R'R''}^{(\ell_z \oplus r', \tilde{s})} \ket{r', \ell_x}_{R'R''}$$

$$= \sum_r \ket{\bar{r}}\bra{\overline{r \oplus \ell_z}, \tilde{s}} (-1)^{(r\oplus\ell_z)\ell_x} \cancel{U_{R'R''}^{r \leftarrow r \oplus \ell_x} V_{R'R''}^{(\ell_z \oplus r, \tilde{s})\dagger}}$$

$$p(\ell, \tilde{s}) \sum_{r'} \ket{\overline{\ell_z \oplus r'}, \tilde{s}}\bra{\overline{\ell_z \oplus r'}, \tilde{s}}_Q \cancel{V_{R'R''}^{(\ell_z \oplus r', \tilde{s})} \ket{r', \ell_x}_{R'R''}}$$

$$\cancel{\ket{\eta}_{R'R''}}$$

Also have $r = r'$:

$$= \frac{p(\ell, \tilde{s})}{\mu_s} \sum_r \ket{\bar{r}}\bra{\overline{r \oplus \ell_z}, \tilde{s}}_Q \ket{\eta}_{R'R''} (-1)^{(r\oplus\ell_z)\ell_x}$$

So after tracing over $R'R''$, we have:

$$\mathcal{R}_Q = \sum_{\ell\tilde{s}} \left( \sqrt{q(\ell|\tilde{s})} \sum_r \ket{\bar{r}}\bra{\overline{r \oplus \ell_z}, \tilde{s}} (-1)^{(r\oplus\ell_z)\ell_x} \right) (\cdot)(\text{h.c.}),$$

where

$$q(\ell|\tilde{s}) := \frac{p(\ell, \tilde{s})^2}{\sum_{\ell'} p(\ell, \tilde{s})^2}. \tag{S144}$$

Finally, on

$$\sum_r \ket{\bar{r}}\bra{\overline{r \oplus \ell_z}, \tilde{s}} (-1)^{(r\oplus\ell_z)\ell_x}, \tag{S145}$$

We can shift $r \to r \oplus \ell_z$ and use the same trick as the Pauli errors and Petz case to write it as $\bar{X}^{\ell_z}\bar{Z}^{\ell_x}(ZX)^{\tilde{s}}\Pi_s$, to get

$$\mathcal{R}_Q = \sum_{\ell\tilde{s}} \left( \sqrt{q(\ell|\tilde{s})} \bar{X}^{\ell_z}\bar{Z}^{\ell_x}(ZX)^{\tilde{s}}\Pi_s \right)(\cdot)(\text{h.c.}),$$

which is a measurement of syndromes followed by a correction step from a sampling decoder with probabilities $q(\ell|\tilde{s})$.

### D. Amplitude damping and the Petz recovery channel

The Kraus operator for the entire channel of amplitude damping channel can be written as

$$K_\sigma = \sum_\tau \sqrt{p_{\sigma\tau}} Z^\tau X^\sigma. \tag{S146}$$

The Pauli string acts as

$$Z^\tau X^\sigma \ket{\bar{k}} = (-1)^{f(\sigma,\tau)+\ell_x(\tau)(k\oplus\ell_z(\sigma))} \ket{\overline{k \oplus l_z(\sigma)}, \tilde{s}_x(\tau), \tilde{s}_z(\sigma)}. \tag{S147}$$

For ease of notation, from below, we no longer label the syndrome basis with tildes. Then we have

$$\mathcal{E}[\Pi] = \sum_{k,\ell_z,s_z} \sum_{\sigma \in \ell_z, s_z} \sum_{\ell_x s_x} \sum_{\tau \in \ell_x, s_x} \sqrt{p_{\sigma\tau}}(-1)^{f(\sigma,\tau)+\ell_x(k\oplus\ell_z)} \ket{\overline{k \oplus \ell_z}, s_x, s_z} \sum_{\ell'_x s'_x} \sum_{\tau' \in \ell'_x, s'_x} \bra{\overline{k \oplus \ell_z}, s'_x, s_z} \sqrt{p_{\sigma\tau'}}(-1)^{f(\sigma,\tau')+\ell'_x(k\oplus\ell_z)}. \tag{S148}$$

Let $\alpha_{\ell_x s_x \sigma} := \sum_{\tau \in \ell_x s_x} \sqrt{p_{\sigma\tau}} (-1)^{f(\sigma,\tau)}$:

$$\mathcal{E}[\Pi] = \sum_{k,\ell_z,s_z} \sum_{\sigma \in \ell_z,s_z} \sum_{\ell_x s_x} \alpha_{\ell_x s_x \sigma} (-1)^{\ell_x(k\oplus\ell_z)} \left|\overline{k\oplus\ell_z}, s_x, s_z\right\rangle \sum_{\ell'_x s'_x} \left\langle\overline{k\oplus\ell_z}, s'_x, s_z\right| \alpha_{\ell'_x s'_x \sigma} (-1)^{\ell'_x(k\oplus\ell_z)}. \tag{S149}$$

Summing over $\sigma$, define $A^{\ell_z s_z}_{\ell_x s_x, \ell'_x s'_x} = \sum_{\sigma \in \ell_z s_z} \alpha_{\ell_x s_x \sigma} \alpha_{\ell'_x s'_x \sigma}$:

$$\mathcal{E}[\Pi] = \sum_{k,\ell_z,s_z} \sum_{\ell_x s_x} \sum_{\ell'_x s'_x} (-1)^{\ell_x(k\oplus\ell_z)} \left|\overline{k\oplus\ell_z}, s_x, s_z\right\rangle A^{\ell_z s_z}_{\ell_x s_x, \ell'_x s'_x} \left\langle\overline{k\oplus\ell_z}, s'_x, s_z\right| (-1)^{\ell'_x(k\oplus\ell_z)}. \tag{S150}$$

Note that we can shift the sum $k \to k \oplus \ell_z$, which gives

$$\mathcal{E}[\Pi] = \sum_{k,\ell_z,s_z} \sum_{\ell_x s_x} \sum_{\ell'_x s'_x} (-1)^{\ell_x k} \left|\overline{k}, s_x, s_z\right\rangle A^{\ell_z s_z}_{\ell_x s_x, \ell'_x s'_x} \left\langle\overline{k}, s'_x, s_z\right| (-1)^{\ell'_x k}. \tag{S151}$$

We may now sum over $\ell_z$, leaving $A^{s_z}_{\ell_x s_x, \ell'_x s'_x} := \sum_{\ell_z} A^{\ell_z s_z}_{\ell_x s_x, \ell'_x s'_x}$:

$$\mathcal{E}[\Pi] = \sum_{k,s_z} \sum_{\ell_x s_x} \sum_{\ell'_x s'_x} (-1)^{\ell_x k} \left|\overline{k}, s_x, s_z\right\rangle A^{s_z}_{\ell_x s_x, \ell'_x s'_x} \left\langle\overline{k}, s'_x, s_z\right| (-1)^{\ell'_x k}. \tag{S152}$$

Let $A^{k s_z}_{s_x, s'_x} := \sum_{\ell_x \ell'_x} (-1)^{\ell_x k} A^{s_z}_{\ell_x s_x, \ell'_x s'_x} (-1)^{\ell'_x k}$:

$$\mathcal{E}[\Pi] = \sum_{k,s_z} \sum_{s_x s'_x} \left|\overline{k}, s_x, s_z\right\rangle A^{k s_z}_{s_x, s'_x} \left\langle\overline{k}, s'_x, s_z\right|. \tag{S153}$$

Now formally diagonalise this matrix: $A^{k s_z}_{s_x, s'_x}$. It is a real, symmetric matrix and therefore has real eigenvalues. Let us write the eigenvectors as $\left|v^{k s_z a}\right\rangle = \sum_{s_x} v^{k s_z a}_{s_x} \left|s_x\right\rangle$. Then we have

$$\mathcal{E}[\Pi] = \sum_{k,s_z,a} \lambda_{k s_z a} \left|\overline{k}, v^{k s_z a}, s_z\right\rangle\!\left\langle\overline{k}, v^{k s_z a}, s_z\right|. \tag{S154}$$

Now let us consider the recovery Kraus operator

$$R_\sigma = \Pi K^\dagger_\sigma \mathcal{E}(\Pi)^{-1/2} = \sum_k \left|\overline{k}\right\rangle\!\left\langle\overline{k}\right| \sum_\tau X^\sigma Z^\tau \sqrt{p_{\sigma\tau}} \sum_{k',s_z,a} \lambda^{-1/2}_{k' s_z a} \left|\overline{k'}, v^{k' s_z a}, s_z\right\rangle\!\left\langle\overline{k'}, v^{k' s_z a}, s_z\right| \tag{S155}$$

$$= \sum_k \sum_{\ell_x,s_x} \left|\overline{k}\right\rangle\!\left\langle\overline{k\oplus\ell_z(\sigma)}, s_x, s_z(\sigma)\right| \sum_{\tau \in \ell_x,s_x} (-1)^{f(\tau,\sigma)} (-1)^{\ell_x(k\oplus\ell_z(\sigma))} \sqrt{p_{\sigma\tau}} \sum_{k',s_z,a} \lambda^{-1/2}_{k' s_z a} \left|\overline{k'}, v^{k' s_z a}, s_z\right\rangle\!\left\langle\overline{k'}, v^{k' s_z a}, s_z\right| \tag{S156}$$

$$= \sum_k \sum_{\ell_x,s_x} \alpha_{\ell_x s_x \sigma} (-1)^{\ell_x(k\oplus\ell_z(\sigma))} \left|\overline{k}\right\rangle\!\left\langle\overline{k\oplus\ell_z(\sigma)}, s_x, s_z(\sigma)\right| \sum_{k',s_z,a} \lambda^{-1/2}_{k' s_z a} \left|\overline{k'}, v^{k' s_z a}, s_z\right\rangle\!\left\langle\overline{k'}, v^{k' s_z a}, s_z\right| \tag{S157}$$

$$= \sum_k \sum_{\ell_x,s_x} \alpha_{\ell_x s_x \sigma} \sum_{k',s_z,a} \lambda^{-1/2}_{k' s_z a} (-1)^{\ell_x(k\oplus\ell_z(\sigma))} \left|\overline{k}\right\rangle\!\left\langle\overline{k'}, v^{k' s_z a}, s_z\right| \delta_{k\oplus\ell_z(\sigma), k'} \delta_{s_z(\sigma), s_z} \left\langle s_x | v^{k' s_z a}\right\rangle \tag{S158}$$

$$= \sum_k \sum_{\ell_x,s_x} \alpha_{\ell_x s_x \sigma} \sum_a \lambda^{-1/2}_{k\oplus\ell_z(\sigma),s_z(\sigma),a} (-1)^{\ell_x(k\oplus\ell_z(\sigma))} v^{k\oplus\ell_z(\sigma),s_z(\sigma),a}_{s_x} \left|\overline{k}\right\rangle\!\left\langle\overline{k\oplus\ell_z(\sigma)}, v^{k\oplus\ell_z(\sigma),s_z(\sigma),a}, s_z(\sigma)\right|. \tag{S159}$$

Shift $k \to k \oplus \ell_z(\sigma)$ to find

$$R_\sigma = \sum_{k,\ell_x,s_x,a} \alpha_{\ell_x s_x \sigma} \lambda^{-1/2}_{k,s_z(\sigma),a} (-1)^{\ell_x k} v^{k,s_z(\sigma),a}_{s_x} \left|\overline{k\oplus\ell_z(\sigma)}\right\rangle\!\left\langle\overline{k}, v^{k,s_z(\sigma),a}, s_z(\sigma)\right|. \tag{S160}$$

Now consider

$$\mathcal{R}[\cdot] = \sum_{\ell_z s_z} \sum_{\sigma \in \ell_z s_z} R_\sigma \cdot R^\dagger_\sigma. \tag{S161}$$

Performing the sum over $\sigma$, we get $A^{\ell_z s_z}_{\ell_x s_x, \ell'_x s'_x}$:

$$\mathcal{R}[\cdot] = \sum_{\ell_z s_z} \sum_{k, \ell_x, s_x, a} \lambda^{-1/2}_{k, s_z, a} (-1)^{\ell_x k} v^{k, s_z, a}_{s_x} \left| \overline{k \oplus \ell_z} \middle\rangle\!\!\middle\langle \overline{k}, v^{k, s_z, a}, s_z \right| (\cdot) \sum_{k', \ell'_x, s'_x, a'} A^{\ell_z s_z}_{\ell_x s_x, \ell'_x s'_x} (\prime). \tag{S162}$$

To make the structure more transparent, we diagonalise $A^{\ell_z s_z}_{\ell_x s_x, \ell'_x s'_x} = \sum_b \mu_{\ell_z s_z b} u^{\ell_z s_z b}_{\ell_x s_x} u^{\ell_z s_z b *}_{\ell'_x s'_x}$. There are as many $b$ indices as there is $(_x, s_x)$. Next, perform the sum over $\ell_x, s_x$ and gather up various terms:

$$\beta^{\ell_z s_z b}_{ka} := \sum_{\ell_x s_x} \sqrt{\mu_{\ell_z s_z b}} u^{\ell_z s_z b}_{\ell_x s_x} \lambda^{-1/2}_{k, s_z, a} (-1)^{\ell_x k} v^{k, s_z, a}_{s_x}. \tag{S163}$$

Then we can finally write

$$\mathcal{R}[\cdot] = \sum_{\ell_z s_z b} \sum_{k, a} \beta^{\ell_z s_z b}_{ka} \left| \overline{k \oplus \ell_z} \middle\rangle\!\!\middle\langle \overline{k}, v^{k, s_z, a}, s_z \right| (\cdot) \sum_{k' a'} (\prime). \tag{S164}$$

Now expand $\left| v^{k s_z a} \right\rangle$:

$$\mathcal{R}[\cdot] = \sum_{\ell_z s_z b} \sum_{k, a, s_x} \beta^{\ell_z s_z b}_{ka} v^{k, s_z, a}_{s_x} \left| \overline{k \oplus \ell_z} \middle\rangle\!\!\middle\langle \overline{k}, s_x, s_z \right| (\cdot) \sum_{k' a'} (\prime). \tag{S165}$$

Now look at

$$\left| \overline{k \oplus \ell_z} \middle\rangle\!\!\middle\langle \overline{k}, s_x, s_z \right| = \bar{X}^{\ell_z} \left| \overline{k} \middle\rangle\!\!\middle\langle \overline{k} \right| Z^{s_x} X^{s_z} = \bar{X}^{\ell_z} \left| \overline{k} \middle\rangle\!\!\middle\langle \overline{k} \right| Z^{s_x} X^{s_z} X^{s_z} Z^{s_x} \Pi Z^{s_x} X^{s_z} \tag{S166}$$

$$= \bar{X}^{\ell_z} \left| \overline{k} \middle\rangle\!\!\middle\langle \overline{k} \right| Z^{s_x} X^{s_z} \Pi_{s_x, s_z}. \tag{S167}$$

Now $\Pi_{s_x, s_z} = \Pi_{s_x} \Pi_{s_z}$, since the vertex and plaquette projectors (with appropriate signs, i.e. $\mathbb{A}^{s^v_x}_v = \frac{1}{2}(1 + s^v_x X_v)$, etc.) commute. The projector $\Pi_{s_x}$ can be absorbed to the bra on the left. So we have

$$\left| \overline{k \oplus \ell_z} \middle\rangle\!\!\middle\langle \overline{k}, s_x, s_z \right| = \bar{X}^{\ell_z} Z^{s_x} X^{s_z} X^{s_z} Z^{s_x} \left| \overline{k} \middle\rangle\!\!\middle\langle \overline{k} \right| Z^{s_x} X^{s_z} \Pi_{s_x} \Pi_{s_z} = \bar{X}^{\ell_z} Z^{s_x} X^{s_z} \Pi_{k, s_x, s_z} \Pi_{s_z}. \tag{S168}$$

So we have

$$\mathcal{R}[\cdot] = \sum_{\ell_z s_z b} \sum_{k, a, s_x} \beta^{\ell_z s_z b}_{ka} v^{k, s_z, a}_{s_x} \bar{X}^{\ell_z} Z^{s_x} X^{s_z} \Pi_{k, s_x, s_z} \Pi_{s_z} (\cdot) \sum_{k' a'} (\prime) \tag{S169}$$

$$= \sum_{\ell_z s_z b} \bar{X}^{\ell_z} X^{s_z} \left[ \sum_{k, a, s_x} \beta^{\ell_z s_z b}_{ka} v^{k, s_z, a}_{s_x} (-1)^{g(s_x, s_z)} Z^{s_x} \Pi_{k, s_x, s_z} \right] \Pi_{s_z} (\cdot) \sum_{k' a'} (\prime). \tag{S170}$$

We perform the sum over $a$: $w^{\ell_z s_z b}_{k s_x} := \sum_a \beta^{\ell_z s_z b}_{ka} v^{k, s_z, a}_{s_x} (-1)^{g(s_x, s_z)}$. So we have

$$\mathcal{R}[\cdot] = \sum_{\ell_z s_z b} \bar{X}^{\ell_z} X^{s_z} \left[ \sum_{k, s_x} w^{\ell_z s_z b}_{k s_x} Z^{s_x} \Pi_{k, s_x, s_z} \right] \Pi_{s_z} (\cdot) \sum_{k' s'_x} (\prime). \tag{S171}$$

Let us Fourier transform $w^{\ell_z s_z b}_{k s_x}$ with respect to $k$, to get $\hat{w}^{\ell_z s_z b}_{\ell_x s_x} = \frac{1}{2^K} \sum_k w^{\ell_z s_z b}_{k s_x} (-1)^{k \cdot \ell_x}$, meaning $w^{\ell_z s_z b}_{k s_x} = \sum_{\ell_x} \hat{w}^{\ell_z s_z b}_{\ell_x s_x} (-1)^{k \ell_x}$. Here, $K$ is the number of logical qubits, and it should be reminded that $k$ and $\ell_x$ are vectors of length $K$. Then we have

$$\left[ \sum_{k, s_x} w^{\ell_z s_z b}_{k s_x} Z^{s_x} \Pi_{k, s_x, s_z} \right] = \left[ \sum_{\ell_x, k, s_x} \hat{w}^{\ell_z s_z b}_{\ell_x s_x} (-1)^{k \ell_x} Z^{s_x} \Pi_{k, s_x, s_z} \right] = \left[ \sum_{\ell_x, s_x} \hat{w}^{\ell_z s_z b}_{\ell_x s_x} Z^{s_x} \sum_k (-1)^{k \ell_x} \Pi_{k, s_x, s_z} \right] \tag{S172}$$

$$= \left[ \sum_{\ell_x, s_x} \hat{w}^{\ell_z s_z b}_{\ell_x s_x} Z^{s_x} \bar{Z}^{\ell_x} \Pi_{s_x, s_z} \right]. \tag{S173}$$

This allows us to write the recovery map as

$$\mathcal{R}[\cdot] = \sum_{\ell_z s_z b} \bar{X}^{\ell_z} X^{s_z} \left[ \sum_{\ell_x, s_x} \hat{w}^{\ell_z s_z b}_{\ell_x s_x} Z^{s_x} \bar{Z}^{\ell_x} \Pi_{s_x} \right] \Pi_{s_z} (\cdot) \sum_{\ell'_x s'_x} (\prime). \tag{S174}$$

This proves that $Z$-type stabilisers are measured, as well as the $X$-type string operators as the correction unitaries. For the $X$-type stabilisers, in each sector, a coherent rotation is applied, with phases with weights given by $\hat{w}^{\ell_z s_z b}_{\ell_x s_x}$.

### E. Amplitude damping and the SW recovery channel

For ease of notation, from below, we no longer label the syndrome basis with tildes. We have

$$\left|\Psi'_{RQE}\right\rangle = \sum_k \frac{1}{\sqrt{2}} |k\rangle_R \sum_{\sigma,\tau} \sqrt{p(\sigma,\tau)} |\bar{k},\sigma,\tau\rangle_R |\sigma\rangle_E \tag{S175}$$

$$= \sum_k \frac{1}{\sqrt{2}} |k\rangle_R \sum_{s_z,\ell_z} \sum_{s_x,\ell_x} (-1)^{\ell_x(k\oplus \ell_z)} |\overline{k\oplus \ell_z},s_z,s_x\rangle_Q \sum_{\sigma \in s_z,\ell_z} \sum_{\tau \in s_x,\ell_x} \sqrt{p(\sigma,\tau)}(-1)^{f(\sigma,\tau)} |\sigma\rangle_E. \tag{S176}$$

Let $W(\sigma,s_x,\ell_x) = \sum_{\tau \in s_x,\ell_x} \sqrt{p(\sigma,\tau)}(-1)^{f(\sigma,\tau)}$. Then we have

$$\left|\Psi'_{RQE}\right\rangle = \sum_k \frac{1}{\sqrt{2}} |k\rangle_R \sum_{\sigma,\tau} \sqrt{p(\sigma,\tau)} |\bar{k},\sigma,\tau\rangle_R |\sigma\rangle_E \tag{S177}$$

$$= \sum_k \frac{1}{\sqrt{2}} |k\rangle_R \sum_{s_z,\ell_z} \sum_{s_x,\ell_x} (-1)^{\ell_x(k\oplus \ell_z)} |\overline{k\oplus \ell_z},s_z,s_x\rangle_Q \sum_{\sigma \in s_z,\ell_z} W(\sigma,s_x,\ell_x) |\sigma\rangle_E. \tag{S178}$$

Consider doing the partial trace over $Q$ and $E$. we notice that we will have $k_B = k_K$, as well as $k_B \oplus \ell_B^z = k_K \oplus \ell_K^z$ therefore $\ell_B^z = \ell_K^z$. We also have $s_B^z = s_K^z$ and $s_B^x = s_K^x$. Here, $B$ denotes 'bra' and $K$ denotes 'ket'. So we have

$$\rho'_E = \sum_k \frac{1}{2} \sum_{s_z,\ell_z} \sum_{s_x} \sum_{\ell_x,\ell'_x} (-1)^{\ell_x(k\oplus \ell_z)}(-1)^{\ell'_x(k\oplus \ell_z)} \sum_{\sigma,\sigma' \in s_z,\ell_z} W(\sigma,s_x,\ell_x)W(\sigma,s_x,\ell'_x) |\sigma\rangle\langle \sigma'|_E. \tag{S179}$$

Now $\sum_k \frac{1}{2}(-1)^{\ell_x(k\oplus \ell_z)}(-1)^{\ell'_x(k\oplus \ell_z)} = \delta_{\ell_x,\ell'_x}$, so

$$\rho'_E = \sum_{s_z,\ell_z} \sum_{s_x,\ell_x} \left(\sum_{\sigma \in s_z,\ell_z} W(\sigma,s_x,\ell_x)|\sigma\rangle_E\right) (\text{h.c.}). \tag{S180}$$

Purifying, we have (where $Q' = R'R''Q$)

$$\left|\tilde{\Psi}_{Q'RE}\right\rangle = \left(\sum_m \frac{1}{\sqrt{2}} |m\rangle_{R'} |m\rangle_R\right) \otimes \left(\sum_{s_z,\ell_z,s_x,\ell_x} |\bar{\ell}_z,s_z,s_x\rangle_Q |\ell_x\rangle_{R''} \sum_{\sigma \in s_z,\ell_z} W(\sigma,s_x,\ell_x)|\sigma\rangle_E\right). \tag{S181}$$

We want to compute:

$$M_{Q'} = \text{tr}_{RE}\left[\left|\tilde{\Psi}_{Q'RE}\right\rangle\!\!\left\langle\Psi_{QRE}\right| \langle\eta|_{R'}\langle\tilde{\eta}|_{R''}\right] \tag{S182}$$

$$= \text{tr}_{RE}\left[\sum_m \frac{1}{\sqrt{2}} |m\rangle_{R'}|m\rangle_R \sum_{s_z,\ell_z,s_x,\ell_x} |\bar{\ell}_z,s_z,s_x\rangle_Q |\ell_x\rangle_{R''} \sum_{\sigma \in s_z,\ell_z} W(\sigma,s_x,\ell_x)|\sigma\rangle_E \right. \tag{S183}$$

$$\left. \times \sum_k \frac{1}{\sqrt{2}} \langle k|_R \sum_{s'_z,\ell'_z} \sum_{s'_x,\ell'_x} (-1)^{\ell'_x(k\oplus \ell'_z)} \langle\overline{k\oplus \ell'_z},s'_z,s'_x|_Q \sum_{\sigma' \in s'_z,\ell'_z} W(\sigma',s'_x,\ell'_x)\langle\sigma'|_E \langle\eta\tilde{\eta}|_{R'R''}\right]. \tag{S184}$$

The trace over $R$ will force $m = k$. The trace over $E$ will force not only $\sigma = \sigma'$ but also $s_z = s'_X$ and $\ell_z = \ell'_z$. So we have

$$M_{Q'} = \frac{1}{2} \sum_m |m\rangle_{R'} \sum_{s_z,\ell_z,s_x,\ell_x} |\bar{\ell}_z,s_z,s_x\rangle_Q |\ell_x\rangle_{R''} \sum_{\sigma \in s_z,\ell_z} W(\sigma,s_x,\ell_x)W(\sigma,s'_x,\ell'_x) \tag{S185}$$

$$\sum_{s'_x,\ell'_x} (-1)^{\ell'_x(m\oplus \ell_z)} \langle\overline{m\oplus \ell_z},s_z,s'_x|_Q \langle\eta\tilde{\eta}|_{R'R''}. \tag{S186}$$

Let $C_{s_x\ell_x,s'_x\ell'_x}(s_z,\ell_z) := \sum_{\sigma \in s_z,\ell_z} W(\sigma,s_x,\ell_x)W(\sigma,s'_x,\ell'_x)$ and shift $m \to m \oplus \ell_z$:

$$M_{Q'} = \frac{1}{2} \sum_m \sum_{s_z,\ell_z,s_x,\ell_x} |m \oplus \ell_z\rangle_{R'} |\bar{\ell}_z,s_z,s_x\rangle_Q |\ell_x\rangle_{R''} C_{s_x\ell_x,s'_x\ell'_x}(s_z,\ell_z) \sum_{s'_x,\ell'_x} (-1)^{\ell'_x m} \langle\overline{m},s_z,s'_x|_Q \langle\eta\tilde{\eta}|_{R'R''}. \tag{S187}$$

Shift $\ell_z \to \ell_z \oplus m$:

$$M_{Q'} = \frac{1}{2}\sum_m \sum_{s_z,\ell_z,s_x,\ell_x} |\ell_z\rangle_{R'} |\overline{\ell_z \oplus m}, s_z, s_x\rangle_Q |\ell_x\rangle_{R''} C_{s_x\ell_x, s'_x\ell'_x}(s_z, \ell_z \oplus m) \sum_{s'_x,\ell'_x} (-1)^{\ell'_x m} \langle \overline{m}, s_z, s'_x|_Q \langle \eta\tilde{\eta}|_{R'R''} \quad \text{(S188)}$$

$$= \frac{1}{2}\text{CNOT}_{R' \to Q} \sum_m \sum_{s_z,\ell_z,s_x,\ell_x} |\ell_z\rangle_{R'} |\overline{m}, s_z, s_x\rangle_Q |\ell_x\rangle_{R''} C_{s_x\ell_x, s'_x\ell'_x}(s_z, \ell_z \oplus m) \sum_{s'_x,\ell'_x} (-1)^{\ell'_x m} \langle \overline{m}, s_z, s'_x|_Q \langle \eta\tilde{\eta}|_{R'R''}. \quad \text{(S189)}$$

Define

$$|\tilde{w}^{m,s_z}_{s_x,s'_x}\rangle_{R'R''} = \frac{1}{2}\sum_{\ell_z,\ell_x,\ell'_x} C_{s_x\ell_x, s'_x\ell'_x}(s_z, \ell_z \oplus m)(-1)^{\ell'_x m} |\ell_z\rangle_{R'} |\ell_x\rangle_{R''}. \quad \text{(S190)}$$

Then we have

$$M_{Q'} = \text{CNOT}_{R' \to Q} \sum_m \sum_{s_z,s_x,s'_x} |\overline{m}, s_z, s_x\rangle_Q |\tilde{w}^{m,s_z}_{s_x,s'_x}\rangle_{R'R''} \langle \overline{m}, s_z, s'_x|_Q \langle \eta\tilde{\eta}|_{R'R''} \quad \text{(S191)}$$

$$= \text{CNOT}_{R' \to Q} \sum_{m,s_z} |\overline{m}\rangle\langle\overline{m}|_{Q_L} |s_z\rangle\langle s_z|_{Q_Z} \sum_{s_x,s'_x} |s_x\rangle\langle s'_x|_{Q_X} |\tilde{w}^{m,s_z}_{s_x,s'_x}\rangle\langle \eta\tilde{\eta}|_{R'R''}. \quad \text{(S192)}$$

Now the diagonalisation of $M^{m,s_z} = \sum_{s_x,s'_x} |s_x\rangle\langle s'_x|_{Q_X} |\tilde{w}^{m,s_z}_{s_x,s'_x}\rangle\langle \eta\tilde{\eta}|_{R'R''}$ is non-trivial, due to the sign-structure of the matrix. However, let us assume that we are able to perform SVD on this matrix as follows:

$$M^{m,s_z} = \sum_j \nu_j^{m,s_z} |u_j\rangle\langle v_j|_{Q_X} \langle \eta\tilde{\eta}|_{R'R''}. \quad \text{(S193)}$$

Then the optimal unitary is given by

$$U_{Q'}^{\text{argmax}} = \sum_{m,s_z,j} |v_j\rangle_{Q_X} |\eta\tilde{\eta}\rangle_{R'R''} \langle u_j|_{Q_X R'R''} |\overline{m}\rangle\langle\overline{m}|_{Q_L} |s_z\rangle\langle s_z|_{Q_{C_X}} \text{CNOT}_{R' \to Q}. \quad \text{(S194)}$$

Now applying this onto $|\tilde{\Psi}_{Q'RE}\rangle$, we get

$$\sum_{m,s_z,j} |v_j\rangle_{Q_X} |\eta\tilde{\eta}\rangle_{R'R''} \langle u_j|_{Q_X R'R''} |\overline{m}\rangle\langle\overline{m}|_{Q_L} |s_z\rangle\langle s_z|_{Q_{C_X}} \text{CNOT}_{R' \to Q} \quad \text{(S195)}$$

$$\times \sum_{m'} \frac{1}{\sqrt{2}} |m'\rangle_{R'} |m'\rangle_R \sum_{s'_z,\ell'_z,s'_x,\ell'_x} |\overline{\ell'_z}, s'_z, s'_x\rangle_Q |\ell'_x\rangle_{R''} \sum_{\sigma \in s'_x, \ell'_z} W(\sigma, s'_x, \ell'_x) |\sigma\rangle_E \quad \text{(S196)}$$

$$= \sum_{m,s_z,j} |v_j\rangle_{Q_X} |\eta\tilde{\eta}\rangle_{R'R''} \langle u_j|s'_x m' \ell'_x\rangle_{Q_X R'R''} |\overline{m}\rangle_{Q_L} |s_z\rangle_{Q_{C_X}} \quad \text{(S197)}$$

$$\times \sum_{m'} \delta_{m,\ell'_z \oplus m'} \frac{1}{\sqrt{2}} |m'\rangle_{R'} |m'\rangle_R \sum_{\ell'_z,s'_x,\ell'_x} |\overline{\ell'_z \oplus m'}, s'_x\rangle_Q |\ell'_x\rangle_{R''} \sum_{\sigma \in s_z,\ell'_z} W(\sigma, s'_x, \ell'_x) |\sigma\rangle_E \quad \text{(S198)}$$

$$= \sum_{m,s_z,\ell_z,j} \frac{1}{\sqrt{2}} |m\rangle_R \underbrace{|\overline{m \oplus \ell_z}, s_z, v_j\rangle_Q |\eta\tilde{\eta}\rangle_{R'R''}}_{|\phi_{Q'}^{m,(s_z,\ell_z,j)}\rangle} \underbrace{\sum_{s_x,\ell_x} \langle u_j|s_x, m, \ell_x\rangle_{Q_X R'R''} \sum_{\sigma \in s_z,\ell_z} W(\sigma, s_x, \ell_x) |\sigma\rangle_E}_{|e_{s_z,\ell_z,j}\rangle_E}. \quad \text{(S199)}$$

Again, from the diagonal structure of $|\phi^{m,(s_z,\ell_z,j)}\rangle_{Q'}$ with respect to $s_z$, we can show that $s_z$ is measured, in a similar manner to the Petz and amplitude damping case.

---



\end{document}